\documentclass[11pt]{amsart}

\usepackage{amssymb,bbm,amscd}
\usepackage{graphicx}
\usepackage{a4wide}

\theoremstyle{plain}
\newtheorem{theorem}{Theorem}
\newtheorem{prop}{Proposition}
\newtheorem{lemma}{Lemma}

\theoremstyle{definition}
\newtheorem{remark}{Remark}

\newcommand{\dd}{\,\mathrm{d}}
\newcommand{\ts}{\hspace{0.5pt}}

\newcommand{\ZZ}{\mathbb{Z}}
\newcommand{\RR}{\mathbb{R}\ts}
\newcommand{\CC}{\mathbb{C}\ts}
\newcommand{\NN}{\mathbb{N}}
\newcommand{\XX}{\mathbb{X}}
\newcommand{\YY}{\mathbb{Y}}

\newcommand{\exend}{\hfill $\Diamond$}

\newcommand{\myfrac}[2]{\frac{\raisebox{-2pt}{$#1$}}{\raisebox{0.5pt}{$#2$}}}
  
\begin{document}

\title[Generalised Thue-Morse sequences]{Spectral and topological
properties of \\[2mm] a family of generalised Thue-Morse sequences}

\author{Michael Baake}
\author{Franz G\"{a}hler}

\address{Fakult\"{a}t f\"{u}r Mathematik, Universit\"{a}t Bielefeld,\newline
\hspace*{\parindent}Postfach 100131, 33501 Bielefeld, Germany}
\email{$\{$mbaake,gaehler$\}$@math.uni-bielefeld.de}

\author{Uwe Grimm}
\address{Department of Mathematics and Statistics,
The Open University,\newline 
\hspace*{\parindent}Walton Hall, Milton Keynes MK7 6AA, United Kingdom}
\email{u.g.grimm@open.ac.uk}

\begin{abstract}
  The classic middle-thirds Cantor set leads to a singular continuous
  measure via a distribution function that is know as the Devil's
  staircase. The support of the Cantor measure is a set of zero
  Lebesgue measure. Here, we discuss a class of singular continuous
  measures that emerge in mathematical diffraction theory and lead to
  somewhat similar distribution functions, yet with significant
  differences.  Various properties of these measures are derived. In
  particular, these measures have supports of full Lebesgue measure
  and possess strictly increasing distribution functions. In this
  sense, they mark the opposite end of what is possible for singular
  continuous measures.

  For each member of the family, the underlying dynamical system
  possesses a topological factor with maximal pure point spectrum, and
  a close relation to a solenoid, which is the Kronecker factor of the
  system. The inflation action on the continuous hull is sufficiently
  explicit to permit the calculation of the corresponding dynamical
  zeta functions. This is achieved as a corollary of analysing the
  Anderson-Putnam complex for the determination of the cohomological
  invariants of the corresponding tiling spaces.
\end{abstract}

\maketitle
\thispagestyle{empty}

\centerline{Dedicated to Robert V.\ Moody on the occasion of
          his 70th birthday}

\section{Introduction}

The probably most widely known singular continuous measure emerges as
the unique invariant probability measure for the iterated function
system \cite{Hutchinson} of the classic middle-thirds Cantor set.  The
construction and the distribution function $F$ of the resulting
measure are illustrated in Figure~\ref{fig:cantor}. Due to its shape,
$F$ is known as the Devil's staircase. It is a non-decreasing
continuous function that is constant almost everywhere, which
corresponds to the fact that the underlying measure gives no weight
to single points, but is concentrated on an uncountable set of zero
Lebesgue measure (the Cantor set). The Cantor measure is thus both
continuous and singular, hence purely singular continuous.

\begin{figure}
\begin{center}
\includegraphics[width=0.8\textwidth]{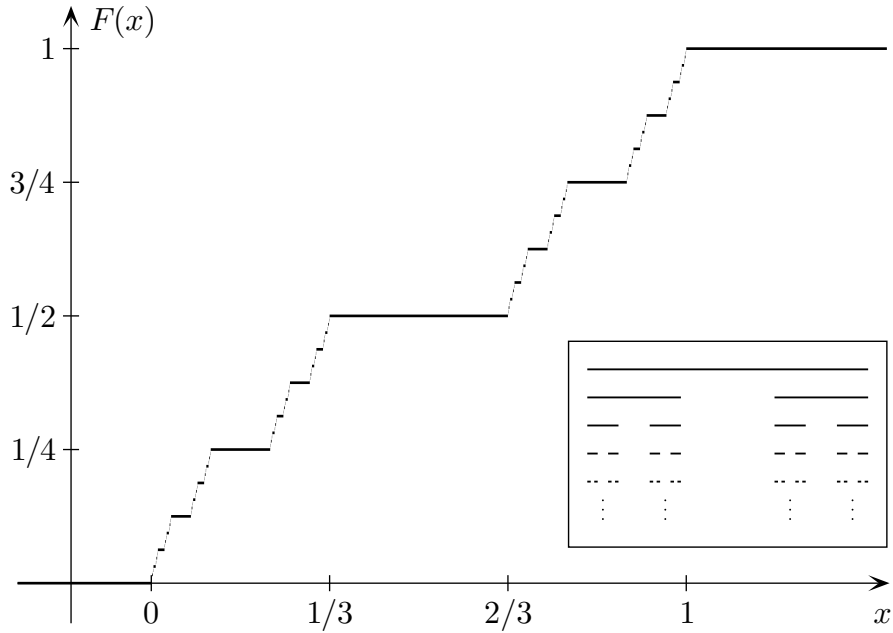}
\end{center}
\caption{\label{fig:cantor}The distribution function $F$ of the classic 
middle-thirds Cantor measure. The construction of the Cantor set is 
sketched in the inset.}
\end{figure}

Singular continuous measures occur in a wide range of physical
problems, most notably in the theory of non-periodic Schr\"{o}dinger
operators; see \cite{BHZ,David} and references therein for
examples. In particular, it is an amazing result that singular
continuous spectra are in a certain sense even generic here; compare
\cite{S,LS}.  One would also expect the appearance of singular
continuous measures in mathematical diffraction theory
\cite{Cowley,Hof,BGrev}, where the Thue-Morse sequence provides one of
the few really explicit examples. Recent experimental evidence
\cite{Withers} indicates that this spectral type might indeed be more
relevant to diffraction than presumed so far. This case has not yet
received the theoretical attention it deserves, though partial results
exist in the dynamical systems literature; compare \cite{Q,Kaku}.

The Thue-Morse system is an example of a bijective substitution of
constant length \cite{Q}. This class has a natural generalisation to
higher dimensions, and is studied in some detail in \cite{Nat1,Nat2};
see also \cite{Zaks} and references therein for related numerical
studies.  Bijective substitutions form an important case within the
larger class of lattice substitution systems. For the latter, from the
point of view of diffraction theory, a big step forward was achieved
in \cite{LM,LMS,LMS2}, where known criteria for pure pointedness in
one dimension \cite{D} were generalised to the case of
$\ZZ^{d}$-action.  Moreover, a systematic connection with model sets
was established (see \cite{M97,Moody00} for detailed expositions and
\cite{Lee} for a rather complete picture), and there are also explicit
algorithms to handle such cases; compare \cite{FS,AL} and references
therein.  Nevertheless, relatively little has been done for the case
without any coincidence in the sense of Dekking \cite{D} or its
generalisation to lattice substitution systems \cite{LM,LMS}.
Although it is believed that one should typically expect singular
continuous measures for bijective substitutions without coincidence,
explicit examples are rare.

As a first step to improve this situation, we investigate a class of
generalised Thue-Morse sequences in the spirit of \cite{Kea}. They are
defined by primitive substitution rules and provide a two-parameter
family of systems with purely singular continuous diffraction. Below,
we formulate a rigorous approach that is constructive and follows the
line of ideas that was originally used by Wiener \cite{Wie}, Mahler
\cite{Mah} and Kakutani \cite{Kaku} for the treatment of the classic
Thue-Morse case. Some of the measures were studied before (mainly by
scaling arguments and numerical methods) in the context of dimension
theory for correlation measures; compare \cite{Kea,Zaks} and
references therein.

The paper is organised as follows. We begin with a brief summary of
the Thue-Morse sequence with its spectral and topological properties,
where we also introduce our notation. Section~\ref{sec:gtm} treats the
family of generalised Thue-Morse sequences from \cite{BBGG}, where the
singular continuous nature of the diffraction spectra is proved and
the corresponding distribution functions are derived. Here, we also
briefly discuss the connection with a generalisation of the period
doubling sequence. The latter has pure point spectrum, and is a
topological factor of the generalised Thue-Morse sequence. This factor
has maximal pure point spectrum.  The diffraction measure of the
generalised Thue-Morse system is analysed in detail in
Section~\ref{sec:gtm-diffract}, and its Riesz product structure is
derived.  In Section~\ref{sec:topol}, we construct the continuous
hulls of the generalised Thue-Morse and period doubling sequences as
inverse limits of the substitution acting on the Anderson-Putnam cell
complex \cite{AP}, and employ this construction to compute and relate
their \v{C}ech cohomologies.  The substitution action on the \v{C}ech
cohomology is then used in Section~\ref{sec:zeta} to derive the
dynamical zeta functions of the corresponding substitution dynamical
systems.  Finally, we conclude with some further observations and open
problems.

\section{A summary of the classic Thue-Morse sequence}\label{sec:tm}

The classic Thue-Morse (or Prouhet-Thue-Morse, abbreviated as TM
below) sequence \cite{A} can be defined via a fixed point of the
primitive substitution
\begin{equation} \label{eq:tm-def}
   \varrho = \varrho^{}_{\mathrm{TM}} \! : \quad
   1 \mapsto 1 \bar{1} \, , \quad \bar{1} \mapsto \bar{1} 1
\end{equation}
on the binary alphabet $\{ 1,\bar{1}\}$. The one-sided fixed point
starting with $1$ reads
\begin{equation}   \label{eq:tm-fix}
    v = v^{}_{0} v^{}_{1} v^{}_{2} \ldots = 
    1 \bar{1} \bar{1} 1 \bar{1} 1 1 \bar{1} \ldots \ts ,
\end{equation}
while $\bar{v}$ is the fixed point starting with $\bar{1}$. One can
now define a two-sided sequence $w$ by
\begin{equation} \label{eq:symm-fix}
    w_i  =  \begin{cases}   v_i ,  & i\ge 0 , \\
                v_{-i-1}, & i<0 .  \end{cases}
\end{equation}
It is easy to check that $w = \ldots w^{}_{-2} w^{}_{-1} | w^{}_{0}
w^{}_{-1} \ldots = \ldots v^{}_{1} v^{}_{0} | v^{}_{0} v^{}_{1}
\ldots$ defines a $2$-cycle under $\varrho$, and hence a fixed point
for $\varrho^{2}$, with the central seed $1|1$ being legal. Recall
that a finite word is called \emph{legal} when it occurs in the
$n$-fold substitution of a single letter for some $n\in\NN$. An
iteration of $\varrho^{2}$ applied to this seed converges to $w$ in
the product topology, which is thus a two-sided fixed point of
$\varrho^{2}$ in the proper sense.

The sequence $w$ defines a dynamical system (under the action of the
group $\ZZ$) as follows. Its compact space is the (discrete)
\emph{hull}, obtained as the closure of the $\ZZ\ts$-orbit of $w$,
\[
      \XX^{\mathrm{TM}}_{} \, = \, 
     \overline{ \{ S^{i} w \mid i\in\ZZ\} }\ts ,
\]
where $S$ denotes the shift operator (with $(Sw)_{i} = w_{i+1}$) and
where the closure is taken in the local (or product) topology. Here,
two sequences are close when they agree on a large segment around the
origin (marked by $|$). Now, $(\XX^{\mathrm{TM}}_{}, \ZZ)$ is a
strictly ergodic dynamical system (hence uniquely ergodic and minimal
\cite{Q,W}). Its unique invariant probability measure is given via the
(absolute) frequencies of finite words (or patches) as the measures of
the corresponding cylinder sets, where the latter then generate the
(Borel) $\sigma$-algebra. Its minimality follows from the repetitivity
of the fixed point word $w$, which also implies that
$\XX^{\mathrm{TM}}_{} = \mathrm{LI} (w)$, which is the local
indistinguishability class of $w$. The latter consists of all elements
of $\{ 1,\bar{1} \}^{\ZZ}$ that are locally indistinguishable from
$w$.

Here, we are interested, for a given $w\in\XX^{\mathrm{TM}}_{}$, in
the diffraction of the (signed) Dirac comb 
\begin{equation}\label{eq:defcomb}
   \omega \, = \, w \, \delta^{}_{\ZZ}
   \, :=  \sum_{n\in\ZZ} w_n \ts \delta_n\ts ,
\end{equation}
where the symbols $1$ and $\bar{1}$ are interpreted as weights $1$ and
$-1$. This defines a mapping from $\XX^{\mathrm{TM}}_{}$ into the
signed translation bounded measures on $\ZZ$ (or on $\RR$). Since this
mapping is a homeomorphism between $\XX^{\mathrm{TM}}_{}$ and its image,
we use both points of view in parallel without further mentioning.

Given any $w\in\XX^{\mathrm{TM}}_{}$, the \emph{autocorrelation
measure} of the corresponding $\omega$ exists as a consequence of
unique ergodicity. It is defined as the volume-averaged (or Eberlein)
convolution
\[
   \gamma \, = \, \omega\circledast \ts\ts \widetilde{\omega} \, =  
   \lim_{N\to\infty} \frac{\omega^{}_{N} * 
   \widetilde{\omega^{}_{N}}}{2N\!+1} \, ,
\]
where $\omega^{}_{N}$ is the restriction of $\omega$ to $[-N,N]$ and
$\widetilde{\mu}$ is the `flipped-over' version of the measure $\mu$
defined by $\widetilde{\mu}(g):=\overline{\mu (\widetilde{g})}$ for
continuous functions $g$ of compact support, with $\widetilde{g}(x)
=\overline{g(-x)}$.  We use this general formulation to allow for
complex weights later on.  A short calculation shows that the
autocorrelation is of the form
\[
    \gamma \, = \, \sum_{m\in\ZZ} \eta(m) \, \delta_{m}
\]
with the autocorrelation coefficients
\begin{equation} \label{eq:auto-coeff}
    \eta(m) \, = \lim_{N\to\infty} \frac{1}{2N\! +1}
    \sum_{n=-N}^{N} w_{n} \, \overline{w_{n-m}} \ts .
\end{equation}
Note that $\gamma$ applies to \emph{all} sequences of
$\XX^{\mathrm{TM}}_{}$, by an application of the ergodic theorem, as
$\eta (m)$ is the orbit average of the continuous function $w\mapsto
w^{}_{0} \ts \overline{w^{}_{-m}}$ and our dynamical system is
uniquely ergodic.

\begin{remark}[Alternative approach]\label{rem:auto-versus-correlation}
  Without the diffraction context, it is possible to directly define
  the function $\eta\! : \, \ZZ \longrightarrow \CC$ by
  \eqref{eq:auto-coeff}.  It is then a positive definite function on
  $\ZZ$, with a representation as the (inverse) Fourier transform of a
  positive measure $\mu$ on the unit circle, by the Herglotz-Bochner
  theorem \cite{R}.  Which formulation one uses is largely a matter of
  taste.  We follow the route via the embedding as a measure on $\RR$,
  so that we get $\widehat{\gamma} = \mu * \delta^{}_{\ZZ}$ together
  with its interpretation as the diffraction measure of the Dirac comb
  $\omega$.  \exend
\end{remark}

We can now employ the special structure of our fixed point $w$ to
analyse $\gamma$. One finds that $\eta (0) = 1$, $\eta (-m) = \eta
(m)$ for all $m\in\ZZ$ and
\begin{equation} \label{eq:tm-coeff}
     \eta (m) \,  = \lim_{N\to\infty} \frac{1}{N}
        \sum_{n=0}^{N-1} v_{n} v_{n+m}
\end{equation}
for all $m\ge 0$. Here, the structure of $w$ and its relation to $v$
was used to derive \eqref{eq:tm-coeff} from \eqref{eq:auto-coeff}.
Observing that $v$ satisfies $ v^{}_{2n} = v^{}_{n}$ and $v^{}_{2n+1}
= \bar{v}^{}_{n}$ for all $n\ge 0$, one can employ
\eqref{eq:tm-coeff} to infer the linear recursion relations
\begin{equation} \label{eq:tm-rec}
     \eta (2m) \, = \, \eta (m) \quad \mbox{and} \quad
     \eta(2m\!+\!1) \, = \, - \myfrac{1}{2} \bigl( \eta (m) + 
     \eta (m\!+\!1) \bigr) ,
\end{equation}
which actually hold for all $m\in\ZZ$. These well-known relations
\cite{Kaku} will also follow from our more general results in
Section~\ref{sec:gtm} as a special case.  One finds $\eta (\pm 1) =
-1/3$ from solving the recursion for $m=0$ and $m=-1$ with
$\eta(0)=1$, while all other values are then recursively determined.

To analyse the \emph{diffraction measure} $\widehat{\gamma}$ of the TM
sequence (following \cite{Mah,Kaku}), one can start with its pure
point part. Defining $\varSigma (N) = \sum_{n=-N}^{N} \bigl( \eta (n)
\bigr)^{2}$, one derives $\varSigma (4N) \le \frac{3}{2}\varSigma(2N)$
from the recursion \eqref{eq:tm-rec}; see \cite{BG08} for the detailed
estimate needed. This implies $\frac{1}{N} \varSigma (N)
\xrightarrow{\, N\to\infty \,} 0$. By Wiener's criterion \cite{Wie},
this means $\bigl( \widehat{\gamma} \bigr)_{\mathsf{pp}} =0$, so that
$\widehat{\gamma}$ is a continuous measure (see Wiener's Corollary in
\cite[Sec.~I.7.13]{Katz} or Wiener's Lemma in \cite[Sec.\ 4.16]{Nad}
for details).

Defining the (continuous and non-decreasing) distribution function $F$
via $F(x) = \widehat{\gamma} \bigl( [0,x]\bigr)$, another consequence
of \eqref{eq:tm-rec} is the pair of functional relations
\[
   \dd F \bigl( \tfrac{x}{2} \bigr) \pm
   \dd F \bigl(\tfrac{x+1}{2} \bigr) = \biggl\{
   \begin{array}{c} 1 \\\!\!\! - \cos (\pi x)\!\!\! \end{array}
   \biggr\}   \dd F (x) \ts .
\]
Splitting $F$ into its \textsf{sc} and \textsf{ac} parts (which are
unique and must both satisfy these relations) now implies backwards
that the recursion \eqref{eq:tm-rec} holds separately for the two sets
of autocorrelation coefficients, $\eta^{}_{\mathsf{sc}}$ and
$\eta^{}_{\mathsf{ac}}$, with yet unknown initial conditions at
$0$. Since this means $\eta^{}_{\mathsf{ac}}(1) = -\frac{1}{3}\ts
\eta^{}_{\mathsf{ac}}(0)$ together with $\eta^{}_{\mathsf{ac}}(2m) =
\eta^{}_{\mathsf{ac}}(m)$ for all $m\in\NN$, an application of the
Riemann-Lebesgue lemma forces $\eta^{}_{\mathsf{ac}} (0) =0$, and
hence $\eta^{}_{\mathsf{ac}} (m) = 0$ for all $m\in\ZZ$, so that also
$\bigl( \widehat{\gamma} \bigr)_{\mathsf{ac}}=0$; compare
\cite{Kaku}. This shows that $\widehat{\gamma}$ is a singular
measure. With the previous argument, since $\widehat{\gamma} \ne 0$,
we see that it is a purely singular continuous
measure. Figure~\ref{fig1} shows an image, where we have used the
uniformly converging Volterra-type iteration
\[
     F^{}_{n+1} (x) = \myfrac{1}{2} \int_{0}^{2x}
     \bigl( 1 - \cos (\pi y)\bigr) F^{\ts \prime}_{n} (y)
     \dd y  \qquad \mbox{with} \qquad F^{}_{0} (x) = x
\]
to calculate $F$ with sufficient precision (note that $F(x+1) = F(x) +
1$, so that displaying $F$ on $[0,1]$ suffices). In contrast to the Devil's
staircase, the TM function is \emph{strictly} increasing, which means
that there is no plateau (which would indicate a gap in the support
of $\widehat{\gamma}$); see \cite{BG08} and references therein for
details.

\begin{figure}
\begin{center}
\includegraphics[width=0.76\textwidth]{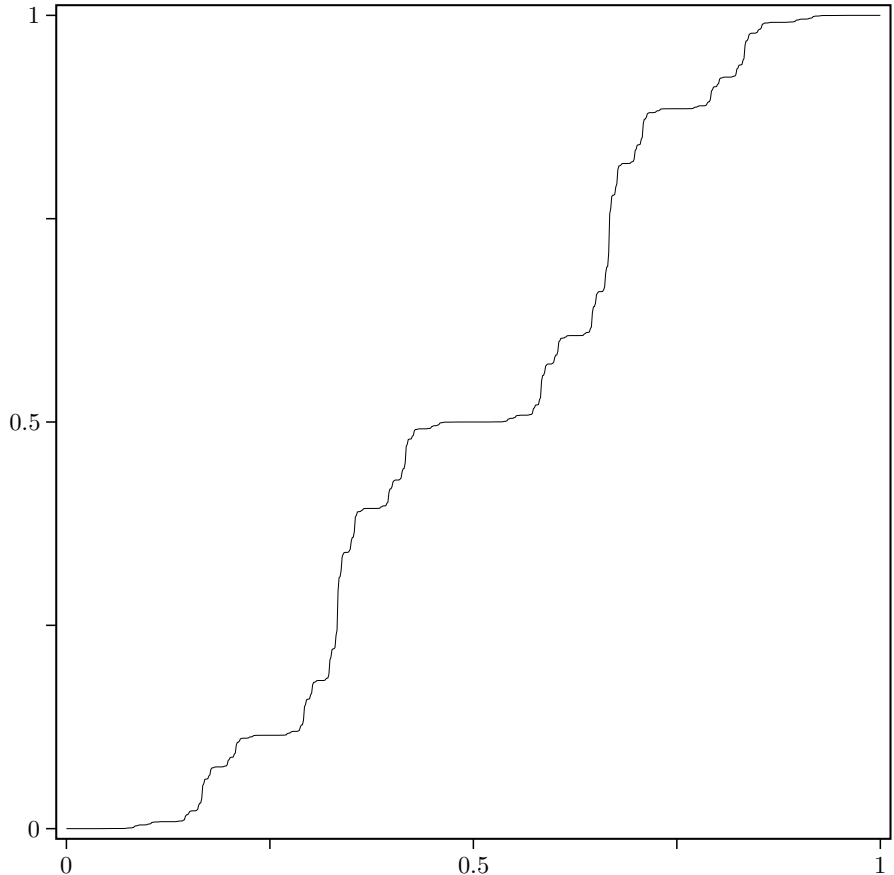}
\end{center}
\caption{\label{fig1}The strictly increasing distribution function of the 
classic, singular continuous TM measure on $[0,1]$.}
\end{figure}

\smallskip
Despite the above result, the TM sequence is closely related to
the period doubling sequence, via the (continuous) block map
\begin{equation} \label{eq:block-map}
  \phi \! : \quad 1\bar{1} \ts , \ts \bar{1}1 \mapsto a
  \, , \quad 11 \ts , \ts \bar{1}\bar{1} \mapsto b \, ,
\end{equation}
which defines an exact 2-to-1 surjection from the hull
$\XX^{\mathrm{TM}}_{}$ to $\XX^{\mathrm{pd}}_{}$, where the latter is
the hull of the period doubling substitution defined by
\begin{equation}\label{eq:pd-def}
   \varrho^{}_{\mathrm{pd}} \! : \quad
   a \mapsto ab \, , \quad b \mapsto aa\ts .
\end{equation}
Viewed as topological dynamical systems, this means that
$(\XX^{\mathrm{pd}}_{},\ZZ)$ is a factor of
$(\XX^{\mathrm{TM}}_{},\ZZ)$. Since both are strictly ergodic, this
extends to the corresponding measure theoretic dynamical systems, as
well as to the suspensions into continuous dynamical systems under the
action of $\RR$.  The dynamical spectrum of the TM system then
contains $\ZZ[\frac{1}{2}]$ as its pure point part \cite{Q}, which is
the entire dynamical spectrum of the period doubling system. We thus
are in the nice situation that a topological factor with maximal pure
point spectrum exists which is itself a substitution system.

The period doubling sequence can be described as a regular model set
with a $2$-adic internal space \cite{BMS,BM} and is thus pure point
diffractive. As another consequence, there is an almost everywhere
$1$-to-$1$ mapping \cite{M,BMS} of the continuous hull (see below)
onto a (dyadic) solenoid $\mathbb{S}= \mathbb{S}_{2}$.  Here, a
solenoid $\mathbb{S}_{m}$ (with $2 \le m\in \NN$ say) is the inverse
limit of the unit circle under the iterated multiplication by $m$. The
dyadic solenoid is obtained for $m=2$.

\medskip

The discrete hull $\XX^{\mathrm{TM}}_{}$ of the TM sequence has a
continuous counterpart (its suspension), which we call
$\YY^{\mathrm{TM}}_{}$. Instead of symbolic TM sequences, one
considers the corresponding tilings of the real line, with labelled
tiles of unit length. Such tilings are not bound to have their
vertices at integer positions, and the full translation group $\RR$
acts continuously on them. The continuous hull of a TM tiling is then
the closure of its $\RR$-orbit with respect to the local
topology. Here, two tilings are close if, possibly after a small
translation, they agree on a large interval around the origin.  For the
same reasons as in the case of the discrete hull, the corresponding
topological dynamical system $(\YY^{\mathrm{TM}}_{},\RR)$ is minimal
and uniquely ergodic, so that every TM tiling defines the same
continuous hull.  Similarly, a continuous hull is defined for the
period doubling sequence.

According to \cite{AP}, the continuous hull of a primitive
substitution tiling can be constructed as the inverse limit of an
iterated map on a finite CW-complex $\Gamma$, called the
Anderson-Putnam (AP) complex. The cells of $\Gamma$ are the tiles in
the tiling, possibly labelled to distinguish different neighbourhoods
of a tile within the tiling. Each point in $\Gamma$ actually
represents a cylinder set of tilings, with a specific neighbourhood at
the origin. The substitution map on the tiling induces a continuous
cellular map of $\Gamma$ onto itself, whose inverse limit is
homeomorphic to the continuous hull of the tiling.

Let $\Gamma_{n}$ be the CW-complex of the $n$th approximation step. It
is a nice feature of the corresponding inverse limit space
$\YY=\varprojlim \Gamma_n$ that its \v{C}ech cohomology $H^k(\YY)$ can
be computed as the direct (inductive) limit of the cohomologies of the
corresponding approximant spaces $H^k(\Gamma_n)$. In our case, we
have a single approximant space $\Gamma$, and a single map (the
substitution map $\varrho$) acting on it.  Consequently, $H^k(\YY)$ is
the inductive limit of the induced map $\varrho^{*}$ on
$H^k(\Gamma)$. Analogous inverse limit constructions also exist for
the hull of the period doubling tiling, $\YY^{\mathrm{pd}}_{}$, and for the
dyadic solenoid, $\mathbb{S}$.

As a consequence of the above, there is a $2$-to-$1$ cover $\phi \! :
\, \YY^{\mathrm{TM}}_{} \rightarrow \YY^{\mathrm{pd}}_{}$, and a
surjection $\psi \! : \, \YY^{\mathrm{pd}}_{} \rightarrow \mathbb{S}$
which is $1$-to-$1$ almost everywhere.  These maps induce well-defined
cellular maps on the associated AP complexes; see also \cite{Sadun}
for a general exposition.  We represent these maps by the same
symbols, $\phi$ and $\psi$.  They induce homomorphisms on the
cohomologies of the AP complexes, so that we have the following
commutative diagram:
\begin{equation} \label{eq:diagram-approx-tm}
   \begin{CD}
    H^k(\Gamma^{\mathrm{sol}}_{})   @>\psi^*>> H^k(\Gamma^{\mathrm{pd}}_{}) 
    @>\phi^*>> H^k(\Gamma^{\mathrm{TM}}_{}) \\
    @V \times 2 VV @V\varrho^{\ts\prime*}VV @VV \varrho^* V \\
    H^k(\Gamma^{\mathrm{sol}}_{})   @>\psi^*>> H^k(\Gamma^{\mathrm{pd}}_{}) 
    @>\phi^*>> H^k(\Gamma^{\mathrm{TM}}_{}) \\
   \end{CD}
\end{equation}
All these maps are explicitly known. The inductive limits along the
columns not only give the cohomologies of the continuous hulls, but
also determine the embeddings under the maps $\phi^*$ and $\psi^*$.
Although $H^1(\YY^{\mathrm{pd}}_{})$ and $H^1(\YY^{\mathrm{TM}}_{})$
are isomorphic, the former embeds (under $\phi^*$) in the latter as a
subgroup of index $2$, which reflects $\YY^{\mathrm{TM}}_{}$ being a
two-fold cover of $\YY^{\mathrm{pd}}_{}$. Furthermore, we get
$H^1(\YY^{\mathrm{pd}}_{})/\psi^*(H^1(\mathbb{S}))=\ZZ$. By an
application of \cite[Prop.~4]{BS}, compare also \cite[Ex.~7]{BS}, this
corresponds to the fact that there are exactly two orbits on which the
map $\psi$ fails to be $1$-to-$1$. These two orbits are merged into a
single orbit under $\psi$.

The action of the substitution on the cohomology of the AP complex
$H^k(\Gamma)$, more precisely the eigenvalues of this action, can be
used \cite{AP} to calculate the dynamical zeta function of the
substitution dynamical system, thus establishing a connection between
the action of the substitution on $H^k(\Gamma)$ and the number of
periodic orbits of the substitution in the continuous hull.  We skip
further details at this point because they will appear later as a
special case of our two-parameter family, which we discuss next.

\section{A family of generalised Thue-Morse sequences}\label{sec:gtm}

The TM sequence is sometimes considered as a rather special and possibly
rare example, which is misleading. In fact, there are many systems
with similar properties.  Let us demonstrate this point by following
\cite{Kea,BBGG}, where certain generalisations were introduced. In
particular, we consider the generalised Thue-Morse (gTM) sequences
defined by the primitive substitutions
\begin{equation} \label{eq:gen-tm-def}
   \varrho = \varrho^{}_{k,\ell} \! : \quad
   1 \mapsto 1^{k} \bar{1}^{\ell} \, , \quad
   \bar{1} \mapsto \bar{1}^{k} 1^{\ell}
\end{equation}
for arbitrary $k,\ell\in\NN$. Here, the one-sided fixed point
starting with $v^{}_{0}=1$ satisfies
\begin{equation} \label{eq:gen-tm-seq}
  v^{}_{m(k+\ell)+r} = \begin{cases}
  v^{}_{m} , & \mbox{if } 0\le r < k \ts , \\
  \bar{v}^{}_{m} , & \mbox{if } k\le r < k\!+\!\ell \ts ,
  \end{cases}
\end{equation}
for $m\ge 0$, as can easily be verified from the fixed point property.
A two-sided gTM sequence $w$ can be constructed as above in
Eq.~\eqref{eq:symm-fix}. Also analogous is the construction of the
discrete hull $\XX^{}_{k,\ell}$ as the orbit closure of $w$ under the
$\ZZ\ts$-action of the shift, where $\XX^{}_{1,1} =
\XX^{\mathrm{TM}}_{}$. We will drop the index when it is clear from
the context. 

\begin{prop}\label{prop:hull}
  Consider the substitution rule\/ $\varrho=\varrho^{}_{k,\ell}$ 
  of Eq.~\eqref{eq:gen-tm-def} for arbitrary, but fixed\/
  $k,\ell\in\NN$. The bi-infinite sequence $w$ that is constructed
  from the one-sided sequence $v$ of \eqref{eq:gen-tm-seq} via
  reflection as in Eq.~\eqref{eq:symm-fix} is a fixed point of\/
  $\varrho^{2}$, as is $w' = \varrho\ts (w)$. Both sequences are
  repetitive and locally indistinguishable.

  Moreover, $w$ is an infinite palindrome and defines the discrete
  hull\/ $\XX = \overline{\{S^{i} w\mid i\in\ZZ \}}$ under the action
  of the shift\/ $S$. This hull is reflection symmetric and minimal,
  and defines a strictly ergodic dynamical system. Similarly, when
  turning $w$ into a tiling of\/ $\RR$ by two intervals of length $1$
  that are distinguished by colour, the closure $\YY$ of the
  $\RR$-orbit of this tiling in the local topology defines a dynamical
  system $(\YY,\RR)$ that is strictly ergodic for the $\RR$-action.
\end{prop}

\begin{proof}
  Most claims are standard consequences of the theory of substitution
  dynamical systems \cite{Q}, as applied to $\varrho$. The fixed point
  property includes the relation $\varrho^{2} (w) = w$ together with
  the legality of the central seed $w^{}_{-1} | w^{}_{0}$ around the
  marker. Note that each fixed point of $\varrho^{n}$, with arbitrary
  $n\in\NN$, is repetitive and defines the same hull $\XX$.  The
  latter is minimal due to repetitivity, and consists precisely of the
  LI class of $w$. Since $w$ and $w'$ coincide on the right of the
  marker, but differ on the left in every position, neither can have a
  non-trivial period (this would contradict their local
  indistinguishability). This is an easy instance of the existence of
  distinct proximal elements \cite{BO}.  Consequently, $\XX =
  \mathrm{LI} (w)$ cannot contain \emph{any} element with a
  non-trivial period, so that $w$ and hence $\varrho$ is aperiodic.

  The action of the shift is clearly continuous in the product
  topology.  Unique ergodicity follows from the uniform existence of
  all pattern frequencies (or from linear repetitivity). This means
  that $(\XX,\ZZ)$ defines a topological dynamical system that leads
  to a strictly ergodic dynamical system $(\XX,\ZZ,\nu)$, where the
  unique measure $\nu$ is defined via the frequencies of patches as
  the measures of the corresponding cylinder sets. The claim about the
  extension to the $\RR$-action on $\YY$ follows from the suspension
  of the discrete system, which is easy here because the constant
  length of the substitution $\varrho$ implies that the canonically
  attached tiling is the one described.
\end{proof}

Let us mention in passing that the discrete hull $\XX$ is a Cantor set,
while the local structure of the continuous hull $\YY$ is a product of
an interval with a Cantor set; compare \cite{BHZ} and references
therein.

Since each choice of $k,\ell$ leads to a strictly ergodic dynamical
system, we know that all autocorrelation coefficients (as defined by
Eq.~\eqref{eq:auto-coeff}, with $\bar{1}\, \widehat{=} -1$)
exist. Clearly, we always have $\eta(0)=1$, while several
possibilities exist to calculate $\eta(\pm 1) =
\frac{k+\ell-3}{k+\ell+1}$.

As before, we turn a gTM sequence $w=(w_{n})^{}_{n\in\ZZ}$
into the Dirac comb
\begin{equation}\label{eq:comb}
   \omega \,=\, \sum_{n\in\ZZ} w_{n}\,\delta_{n}\ts ,
\end{equation}
which is a translation bounded measure. Its autocorrelation
is of the form
\begin{equation}\label{eq:gamma}
   \gamma \,=\, \sum_{m\in\ZZ} \eta(m)\,\delta_{m}
\end{equation}
with the coefficients $\eta(m)$, which can alternatively be calculated
via the one-sided fixed point $v$ as in Eq.~\eqref{eq:tm-coeff}. Let
us now derive a recursion for $\eta(m)$. Since this will be the
`golden key' for almost all spectral properties, we provide a detailed proof.

\begin{lemma}\label{lem:gen-tm-rec} 
  Consider the gTM sequence defined by the primitive substitution
  $\varrho$ of\/ \eqref{eq:gen-tm-def}, for fixed\/ $k,\ell \in
  \NN$. When realised as the Dirac comb of Equation~\eqref{eq:comb},
  each element of the corresponding hull\/ $\XX = \XX^{}_{k,\ell}$
  possesses the unique autocorrelation\/ $\gamma$ of \eqref{eq:gamma},
  where the autocorrelation coefficients satisfy\/ $\eta(0)=1$ and the
  linear recursion
\[ 
     \eta \bigl( (k\!+\!\ell) m + r\bigr) = \myfrac{1}{k+\ell}
     \bigl( \alpha^{}_{k,\ell,r}\, \eta(m) +
     \alpha^{}_{k,\ell,k+\ell-r} \, \eta(m\!+\!1) \bigr),
\]
   with $\alpha^{}_{k,\ell,r} = k+\ell-r -2 \min
   (k,\ell,r,k\!+\!\ell\!-\!r)$, valid for all $m\in\ZZ$ and 
   $0\le r < k\!+\!\ell$. In particular, one has
   $\eta\bigl((k\!+\!\ell)m\bigr)=\eta(m)$ for $m\in\ZZ$.
\end{lemma}
\begin{proof}
  Existence and uniqueness of $\gamma$ are a consequence of
  Proposition~\ref{prop:hull}, via an application of the ergodic
  theorem. The support of the positive definite measure $\gamma$ is
  obviously contained in $\ZZ$, so that $\gamma = \eta\ts
  \delta^{}_{\ZZ}$.

  Since $\eta(0)=1$ is immediate from Eq.~\eqref{eq:auto-coeff}, it
  remains to derive the recursion. We begin with $m\ge 0$ and use
  formula~\eqref{eq:tm-coeff}. When $r=0$, one finds
\[
\begin{split}
   \eta\bigl((k\!+\!\ell)m\bigr)\, & = 
   \lim_{N\to\infty}\frac{1}{N}\sum_{n=0}^{N-1} v_{n}\ts v_{n+(k+\ell)m}\, = 
   \lim_{N'\to\infty}\frac{1}{(k+\ell)N'}\sum_{t=0}^{k+\ell-1}
   \sum_{n=0}^{N'-1} v^{}_{(k+\ell)n+t}\, v^{}_{(k+\ell)(n+m)+t}\\
    & \!\!\stackrel{\eqref{eq:gen-tm-seq}}{=}\!\!      
   \lim_{N'\to\infty}\frac{1}{(k+\ell)N'}\sum_{t=0}^{k+\ell-1}
   \sum_{n=0}^{N'-1} \left\{\begin{array}{@{}ll@{}}
    v_{n}\ts v_{n+m} \ts ,& \text{if $0\le t<k$} \\ 
    \bar{v}_{n}\ts \bar{v}_{n+m}\ts , & \text{if $k\le t<k\!+\!\ell$}  
   \end{array}\right\}\\
   &= \lim_{N'\to\infty}\frac{1}{(k+\ell)N'}\sum_{t=0}^{k+\ell-1}
   \sum_{n=0}^{N'-1} v_{n}\ts v_{n+m} \, =\,  \eta(m)\ts ,
\end{split}
\]
  where the penultimate step follows because $\bar{v}_{n}\ts
  \bar{v}_{n+m}= v_{n}\ts v_{n+m}$ due to $v_{i}\in\{\pm 1\}$.

  For general $r$, one proceeds analogously and finds
\begin{equation}\label{eq:gen-case}
   \eta\bigl((k\!+\!\ell)m+r\bigr)\: = 
   \lim_{N'\to\infty}\frac{1}{(k+\ell)N'}\sum_{t=0}^{k+\ell-1}
   \sum_{n=0}^{N'-1} v^{}_{(k+\ell)n+t}\, v^{}_{(k+\ell)(n+m)+r+t}\, .
\end{equation}
 One now has to split the sum over $t$ according to the ten cases
 of Table~\ref{tab:cases}. In each row, the condition for $t$ is
 formulated in such a way that the difference of the bounds gives the
 proper multiplicity of the resulting term to the sum, which may be
 zero in some cases.

\begin{table}
\caption{\label{tab:cases}Conditions on $r$ and $t$ for the splitting
  of the sum in Eq.~\eqref{eq:gen-case}, and the resulting term for the sum
  from the recursion relation 
\eqref{eq:gen-tm-seq}.}\renewcommand{\arraystretch}{1.2}
\centerline{\begin{tabular}{|c|r@{$\:\le t<\:$}l|l|} \hline
    \multicolumn{3}{|c|}{Conditions} &
    \multicolumn{1}{c|}{Term}\\ \hline & $0$ & $k \! - \!r$ &
    $v_{n}\ts v_{n+m}$ \\ & $k\! -\! r$ & $\min(k,k\! +\! \ell \!-
    \!r)$ & $v_{n}\ts \bar{v}_{n+m}$ \\ $0\le r<k$ & $\min(k,k\!
    +\!\ell\!  -\!r) $&$ k$ & $v_{n}\ts v_{n+m+1}$ \\ \cline{2-4} & $k
    $&$ \max(k,k\! +\! \ell\! -\! r)$ & $\bar{v}_{n}\ts \bar{v}_{n+m}$
    \\ & $\max(k,k\! +\! \ell\! -\! r)$&$ k\! +\!\ell$ & 
    $\bar{v}_{n}\ts v_{n+m+1}$ 
    \\ \hline & $0$ & $\min(k,k\! +\!\ell\! -\! r)$ &
    $v_{n}\ts \bar{v}_{n+m}$ \\ & $\min(k,k\! +\! \ell\! -\! r)$ & $k$
    & $v_{n}\ts v_{n+m+1}$ \\ \cline{2-4} $k\le r< k\! +\! \ell$ & $k$
    & $\max(k,k\! +\! \ell\! -\! r)$ & $\bar{v}_{n}\ts \bar{v}_{n+m}$
    \\ & $\max(k,k\! +\! \ell\! -\! r)$ & $2k\! +\! \ell\! -\! r$ &
    $\bar{v}_{n}\ts v_{n+m+1}$ \\ & $2k\!  +\! \ell\! -\! r$ & $ k\!
    +\! \ell $ & $\bar{v}_{n}\ts \bar{v}_{n+m+1}$ \\ \hline
\end{tabular}\renewcommand{\arraystretch}{1}}
\end{table}
  
Observing $\bar{v}_{n} = - v_{n}$, one simply has to add the 
terms of the form $v_{n} \ts v_{n+m}$ with their signed multiplicities,
which contribute to $\eta(m)$, and those of the form $v_{n} \ts
v_{n+m+1}$, which contribute to $\eta(m\!+\!1)$. For instance, when
$0\le r <k$, one finds the multiplicity of $v_{n}\ts v_{n+m}$ as
\[
\begin{split}
   (k\!-\!r)&  - \bigl(\min(k,k\!+\!\ell\!-\!r)-(k\!-\!r)\bigr)
   + \bigl(\max(k,k\!+\!\ell\!-\!r)-k\bigr)\\
 & =\, k+\ell-r-2\bigl(\min(k,k\!+\!\ell\!-\!r)-k+r\bigr) 
    \, =\,  k+\ell-r-2\min(r,\ell)\\
 & =\, k+\ell-r-2\min(k,\ell,r,k\!+\!\ell\!-\!r) \, =\,  \alpha^{}_{k,\ell,r}
\end{split}
\]
where we used $\min(a,b)+\max(a,b)=a+b$ and, in the last line, the
inequality $0\le r <k$. The required denominator $(k+\ell)$ in the
claimed recursion emerges from the splitting as shown above in
Eq.~\eqref{eq:gen-case}.  Likewise, the multiplicity for
$v_{n}\ts v_{n+m+1}$ calculates as
\[
\begin{split}
  \bigl(k-\min(&k,k\!+\!\ell\!-\!r)\bigr) - \bigl((k\!+\!\ell)-
  \max(k,k\!+\!\ell\!-\!r)\bigr)\\
  &=\, -\ell+\max(k,k\!+\!\ell\!-\!r)-\min(k,k\!+\!\ell\!-\!r)\\
 & =\, 2k-r - 2\min(k,k\!+\!\ell\!-\!r) \,=\, r - 2\min(r,\ell)
 \, = \, \alpha^{}_{k,\ell,k+\ell-r}\ts ,
\end{split}
\]
where we used that $\min(r,\ell)= \min(k,\ell,r,k \!+\!
\ell\!-\!r)$ holds in this case.  The remaining cases (for $k\le r<k
\!+\!\ell$) follow from similar calculations; see
Table~\ref{tab:cases}.  This completes the argument for $\eta(m)$ with
$m\ge 0$.

Finally, we know that $\eta(-n)=\eta(n)$ for all $n\in\ZZ$. Let $m<0$
be fixed, so that $m=-\lvert m \rvert$ with $\lvert m \rvert > 0$. If
$r=0$, one simply has
\[
   \eta \bigl( (k\! + \! \ell) m \bigr) \, = \,
   \eta \bigl( (k\! + \! \ell) \lvert m \rvert \bigr) \, = \,
   \eta (\lvert m \rvert ) \, = \, \eta (m) \ts .
\]
When $1 \le r < k\!+\!\ell$, set $m'=\lvert m \rvert - 1$ and $s=
k+\ell-r$, so that $m'\ge 0$ and $1\le s < k\!+\!\ell$. Then, one
finds
\[
\begin{split}
    \eta \bigl( & (k\! + \! \ell) m + r \bigr) \, = \,
     \eta \bigl( (k\! + \! \ell) \lvert m \rvert - r \bigr) \, = \,
   \eta \bigl( (k\! + \! \ell) m' + s \bigr) \\
   & = \, \myfrac{1}{k+\ell} \bigl( \alpha^{}_{k,\ell,s} \,
   \eta(m') + \alpha^{}_{k,\ell,k+\ell-s}\, \eta(m' \! + \!1) \bigr) 
   \, = \,\myfrac{1}{k+\ell} \bigl( \alpha^{}_{k,\ell,r} \,
   \eta(m) + \alpha^{}_{k,\ell,k+\ell-r}\, \eta(m\! + \! 1)\bigr)
\end{split}
\]
due to the reflection symmetry of $\eta$ together with the recursion
for positive arguments.
\end{proof}

Note that the recursion in Lemma~\ref{lem:gen-tm-rec} uniquely determines
all coefficients $\eta\ts (m)$ once $\eta\ts (0)$ is given. Moreover,
the recursion is linear, which will have strong consequences later.

Since, for any given $k,\ell\in\NN$, every member of the corresponding
hull of weighted Dirac combs has the same autocorrelation measure, we
speak, from now on, of the autocorrelation measure of the gTM system
(for parameters $k,\ell \in \NN$). Let us define
$\varSigma(N)=\sum_{n=-N}^{N} \bigl(\eta(n)\bigr)^2$, where we
suppress the parameters $k$ and $\ell$. For $k=\ell=1$, we know a
bound on the growth rate of $\varSigma(N)$, namely $\Sigma(4N)\le
\frac{3}{2} \Sigma(2N)$, from \cite{Kaku,BG08}. For $k + \ell>2$, we
formulate a similar result (with a technically more involved but
structurally slightly simpler proof) as follows.

\begin{lemma}\label{lem:growth}
  Let $k,\ell\in\NN$ with $k\!+\!\ell>2$ be fixed, and let\/ $\eta(n)$
  with\/ $n\in\ZZ$ be the corresponding autocorrelation coefficients
  from Lemma~$\ref{lem:gen-tm-rec}$. Then, there is some positive
  number $q<k\!+\!\ell$ such that $\varSigma \bigl((k\!+\!\ell)N
  \bigr)\le q\varSigma(N)$ for all $N\in\NN$.
\end{lemma} 
\begin{proof}
  Using the recursions of Lemma~\ref{lem:gen-tm-rec}, one finds
\begin{align*}
   \varSigma\bigl((k\!+\!\ell)N\bigr) &= \!\!
   \sum_{n=-(k+\ell)N}^{(k+\ell)N} \!\!\bigl(\eta(n)\bigr)^{2} 
   \;=\; \bigl(\eta\bigl((k\!+\!\ell)N\bigr)\bigr)^{2} + 
   \sum_{r=0}^{k+\ell-1} \sum_{m=-N}^{N-1} 
   \bigl(\eta\bigl((k\!+\!\ell)m+r\bigr)\bigr)^{2} \\
   &= \varSigma(N) \,+\, \frac{1}{(k+\ell)^{2}}\sum_{r=1}^{k+\ell-1} 
   \sum_{m=-N}^{N-1} \bigl(\alpha^{}_{k,\ell,r}\,\eta(m) + 
   \alpha^{}_{k,\ell,k+\ell-r}\,\eta(m\!+\!1)\bigr)^{2}\\
   &\le \frac{\varSigma(N)}{(k+\ell)^{2}}\biggl( (k+\ell)^{2} +\! 
    \sum_{r=1}^{k+\ell-1} \bigl(\alpha_{k,\ell,r}^{2} + 
    \alpha_{k,\ell,k+\ell-r}^{2}\bigr)\biggr) \,+\, 
    \frac{A}{(k+\ell)^{2}}  
\end{align*}
with $A=\bigl(\sum_{m=-N}^{N-1}2\,\lvert\ts\eta(m)\ts\eta(m\!+\!1)
\rvert\bigr) \bigl(\sum_{r=1}^{k+\ell-1}\lvert\alpha^{}_{k,\ell,r}\,
\alpha^{}_{k,\ell,k+\ell-r}\rvert\bigr)$ being a sum of non-negative
terms only. Noting that
\[
   \sum_{m=-N}^{N-1} 2\,\lvert\ts\eta(m)\ts\eta(m\!+\!1)\rvert \,\;\le 
   \sum_{m=-N}^{N-1}\bigl(\eta(m)\bigr)^{2} + \bigl(\eta(m\!+\!1)\bigr)^{2}
   \;\le\; 2\,\varSigma(N)\, ,
\]
one obtains $A\le \varSigma(N) \sum_{r=1}^{k+\ell-1}2\,
\lvert\alpha^{}_{k,\ell,r}\, \alpha^{}_{k,\ell,k+\ell-r}
\rvert$. Employing the binomial formula results in
\[
  \varSigma\bigl((k\!+\!\ell)N\bigr) \;\le\;
  \frac{\varSigma(N)}{(k+\ell)^{2}}\,\biggl((k+\ell)^{2} +\! 
  \sum_{r=1}^{k+\ell-1} \bigl(\lvert\alpha^{}_{k,\ell,r}\rvert +
  \lvert\alpha^{}_{k,\ell,k+\ell-r}\rvert\bigr)^{2}\biggr).
\]
Our claim follows if we show that the term in the large brackets is
smaller than $(k+\ell)^{3}$. For $1\le r\le k+\ell-1$, we know that
$1\le\min(k,\ell,r,k+\ell-r)\le \min(r,k+\ell-r)$, which implies
$\lvert\alpha^{}_{k,\ell,r}\rvert\le k+\ell-r$ and hence
\begin{equation}\label{eq:alsum}
   \lvert\alpha^{}_{k,\ell,r}\rvert + \lvert\alpha^{}_{k,\ell,k+\ell-r}\rvert
   \:\le\: k+\ell\, .
\end{equation}
Since $k + \ell>2$ by assumption, the stronger inequality
$\lvert\alpha^{}_{k,\ell,1}\rvert + \lvert \alpha^{}_{k,\ell,k+\ell-1}
\rvert\le k+\ell-2$ holds for $r=1$, so that at least one term in the
sum is smaller than $k+\ell$. This means that a $q<k\!+\!\ell$ exists
such that $\varSigma\bigl((k\!+\!\ell)N\bigr)\le q\varSigma(N)$ holds
for all $N\ge 1$.
\end{proof}

The recursion derived in Lemma~\ref{lem:gen-tm-rec} can now be used to
show the absence of pure point components (by Wiener's criterion,
which will rely on Lemma~\ref{lem:growth}) as well as that of
absolutely continuous components (by the Riemann-Lebesgue lemma, which
will rely on the special relation $\eta\bigl((k\!+\!\ell)m\bigr)=\eta(m)$
from Lemma~\ref{lem:gen-tm-rec}), thus establishing that each sequence
in this family leads to a signed Dirac comb with purely singular
continuous diffraction.

\begin{theorem} \label{thm:gen-tm}
   Let $k,\ell\in\NN$. The diffraction measure of the gTM
   substitution $\varrho=\varrho^{}_{k,\ell}$ is the Fourier
   transform $\widehat{\gamma}$ of the autocorrelation measure
   $\gamma$ of Lemma~$\ref{lem:gen-tm-rec}$. It is the diffraction
   measure of every element of the hull of weighted Dirac combs for
   $\varrho$.  Moreover, $\widehat{\gamma}$ is purely singular
   continuous.
\end{theorem}
\begin{proof}
  Since the statement is clear for $k=\ell=1$ from
  Section~\ref{sec:tm} together with \cite{Kaku,BG08}, let $k\!+\!\ell>2$
  be fixed. The corresponding autocorrelation is unique by
  Lemma~\ref{lem:gen-tm-rec}. Since it is positive definite by
  construction, its Fourier transform exists \cite{BF}, and then
  applies to each element of the hull again. Since the underlying
  Dirac comb is supported on $\ZZ$, we know from \cite[Thm.~1]{B} that
  $\widehat{\gamma}$ is $\ZZ\ts$-periodic, hence it can be written as
\begin{equation} \label{eq:conv-rep}
    \widehat{\gamma}  =  \mu * \delta^{}_{\ZZ}
    \quad \text{with} \quad \mu = \widehat{\gamma}\ts |^{}_{[0,1)} \ts .
\end{equation}
Here, $\mu$ is a positive measure on the unit interval (which is a
representation of the unit circle here), so that the inverse Fourier
transform $\check{\mu}$, by the Herglotz-Bochner theorem, is a
(continuous) positive definite function on $\ZZ$ (viewed as the dual
group of the unit circle). Since $\gamma = \check{\mu} \,
\delta^{}_{\ZZ}$ by the convolution theorem together with the Poisson
summation formula $\widehat{\delta_{\ZZ}} = \delta_{\ZZ}$, we see that
this function is
\begin{equation}\label{eq:eta-versus-mu}
    \eta (m) \, = \, \int_{0}^{1} e^{2 \pi i m x} \dd \mu (x)
    \, = \, \check{\mu} (m) \ts .
\end{equation}
  Let us now use the recursion for $\eta$ to infer the spectral nature
  of $\mu$ and thus of $\widehat{\gamma}$.

  Lemma~\ref{lem:growth} implies $\frac{1}{N}\Sigma(N)
  \xrightarrow{\,N\to\infty\,} 0$, which means
  $\bigl(\widehat{\gamma}\bigr)_{\mathsf{pp}}=0$ by Wiener's criterion
  \cite{Wie}; see also \cite{Katz,Nad}. We thus know that
  $\widehat{\gamma} = \bigl(\widehat{\gamma}\bigr)_{\mathsf{sc}}+
  \bigl(\widehat{\gamma}\bigr)_{\mathsf{ac}}$ is a continuous measure,
  where the right-hand side is the sum of two positive measures that
  are mutually orthogonal (in the sense that they are concentrated on
  disjoint sets). Each is the Fourier transform of a positive definite
  measure with support $\ZZ$, hence specified by autocorrelation
  coefficients $\eta^{}_{\mathsf{sc}}$ and $\eta^{}_{\mathsf{ac}}$
  which clearly satisfy $\eta^{}_{\mathsf{sc}}(m) +
  \eta^{}_{\mathsf{ac}}(m) = \eta(m)$ for all $m\in\ZZ$. The recursion
  relations for $\eta$ imply a corresponding set of functional
  relations for the non-decreasing and continuous distribution
  function $F$ defined by $F(x) = \widehat{\gamma} \bigl([0,x]\bigr)$
  for $0\le x \le 1$. Due to the orthogonality mentioned above, the
  same relations have to be satisfied by the $\mathsf{ac}$ and
  $\mathsf{sc}$ parts separately. This in turn implies that
  $\eta^{}_{\mathsf{sc}}$ and $\eta^{}_{\mathsf{ac}}$ must both
  satisfy the recursion relations of Lemma~\ref{lem:gen-tm-rec},
  however with a yet undetermined value of $\eta^{}_{\mathsf{ac}}(0)$,
  and $\eta^{}_{\mathsf{sc}} (0) = 1 - \eta^{}_{\mathsf{ac}}(0)$.
  
  The recursion of Lemma~\ref{lem:gen-tm-rec} with $m=0$ and $r=1$ gives
\[
     \eta^{}_{\mathsf{ac}}(1) \,=\, \frac{k+\ell-3}{k+\ell+1}\,
     \eta^{}_{\mathsf{ac}}(0)\ts ,
\]
while $r=0$ leads to $\eta^{}_{\mathsf{ac}}\bigl((k\!+\!\ell)m\bigr)=
\eta^{}_{\mathsf{ac}}(m)$ for all $m\in\ZZ$. Since we have
$\lim_{n\to\infty}\eta^{}_{\mathsf{ac}}(n)=0$ from the
Riemann-Lebesgue lemma, compare \cite{Katz}, we must have
$\eta^{}_{\mathsf{ac}}(m)=0$ for all $m>0$. When $k + \ell>3$,
$\eta^{}_{\mathsf{ac}}(1)=0$ forces $\eta^{}_{\mathsf{ac}}(0)=0$, and
then $\eta^{}_{\mathsf{ac}}(m)=0$ for all $m\in\ZZ$ by the recursion,
which means $\bigl(\widehat{\gamma}\bigr)_{\mathsf{ac}}=0$. When $k +
\ell=3$, we have $\eta^{}_{\mathsf{ac}}(1)=0$, but the recursion
relation for $m=0$ and $r=2$ leads to $\eta^{}_{\mathsf{ac}}(2) =
-\frac{1}{3}\,\eta^{}_{\mathsf{ac}}(0)$, hence again to
$\eta^{}_{\mathsf{ac}}(0)=0$ with the same conclusion.

  As a consequence, $\eta^{}_{\mathsf{sc}}(0)=1$ and
  $\eta=\eta^{}_{\mathsf{sc}}$. We thus have
  $\widehat{\gamma}=\bigl(\widehat{\gamma}\bigr)_{\mathsf{sc}}$ as
  claimed.
\end{proof}

\begin{remark}[\textsf{Diffraction with general weights}]
  If an arbitrary gTM sequence is given, the diffraction of the
  associated Dirac comb with general (complex) weights $h^{}_{\pm}$
  can be calculated as follows. If $h$ is the function defined by
  $h(1)=h^{}_{+}$ and $h(\bar{1})=h^{}_{-}$, one has
\[
   \omega^{}_{h} \, := \, \sum_{n\in\ZZ} h(w_{n})\ts \delta_{n}
   \, = \, \frac{h^{}_{+} \! +h^{}_{-}}{2} \,\delta^{}_{\ZZ} + 
   \frac{h^{}_{+} \! -h^{}_{-}}{2} \, \omega
\]
  with the $\omega$ from Equation~\eqref{eq:comb}. The autocorrelation
  of $\omega^{}_{h}$ clearly exists and calculates as
\[
   \gamma^{}_{h} \, = \, \frac{\lvert h^{}_{+} \! +
   h^{}_{-}\rvert^{2}}{4} \,\delta^{}_{\ZZ} + 
   \frac{\lvert h^{}_{+} \! -h^{}_{-}\rvert^{2}}{4} \, \gamma
\]
  with $\gamma$ from \eqref{eq:gamma}. This follows from
  $\widetilde{\delta^{}_{\ZZ}} = \delta^{}_{\ZZ}$ together with
  $\delta^{}_{\ZZ}\circledast\omega = \delta^{}_{\ZZ} \circledast
  \widetilde{\omega}=0$, which is a consequence of the fact that $1$
  and $\bar{1}$ are equally frequent in all gTM sequences. The
  diffraction is now obtained as
\[
    \widehat{\gamma^{}_{h}} \, = \, \frac{\lvert h^{}_{+} \! +
   h^{}_{-}\rvert^{2}}{4} \,\delta^{}_{\ZZ} + 
   \frac{\lvert h^{}_{+} \! -h^{}_{-}\rvert^{2}}{4} \, \widehat{\gamma}
\]
  by an application of the Poisson summation formula
  $\widehat{\delta^{}_{\ZZ}}=\delta^{}_{\ZZ}$. Since $\widehat{\gamma}$ is
  purely singular continuous, this is a diffraction measure with
  singular spectrum of mixed type.  \exend
\end{remark}

This diffraction does not display the full dynamical spectrum of the
gTM system, which is a well-known phenomenon from the classic TM
system \cite{ME}. In the latter case, this is `rectified' by the
period doubling system as a topological factor. We will return to this
question for the gTM systems in Remark~\ref{rem:Kronecker}.

\section{The diffraction measure of the gTM 
system}\label{sec:gtm-diffract}

Let us consider the diffraction measure in more detail, which we do
via a suitable distribution function $F$ for the (continuous) measure
$\widehat{\gamma}$. This is done as follows. First, we define $F(x) =
\widehat{\gamma} ([0,1]) = \mu([0,x])$ for $x\in [0,1]$.  The
$\ZZ\ts$-periodicity of $\widehat{\gamma}$ together with $\mu([0,1]) =
\eta(0) = 1$ means that $F$ extends to the entire real line via
$F(x+1) = 1 + F(x)$. Moreover, since $\widetilde{\gamma} = \gamma$, we
know that $\widehat{\gamma}$ is reflection symmetric. With $F(0)=0$,
this implies $F(-x)= - F(x)$ on $\RR$, which is our specification of
$F$ in this case.

\begin{prop}\label{prop:distribution}
  Let $k,\ell \in \NN$ be fixed. The distribution function $F$ of the
  corresponding diffraction measure is non-decreasing, continuous,
  skew-symmetric and satisfies the relation $F(x \! + \! 1) = 1 +
  F(x)$ on the real line. Moreover, it possesses the series expansion
\[
    F(x) \,=\,  x + \sum_{m\ge 1} \frac{\eta(m)}{m\pi} \ts
    \sin (2\pi m x) \ts ,
\]
   which converges uniformly on $\RR$. 
\end{prop}
\begin{proof}
  By construction, $F$ is non-decreasing, and is continuous
  by Theorem~\ref{thm:gen-tm}. So, $F(x) - x$ defines a $1$-periodic
  continuous function that is skew-symmetric and the difference of two
  continuous, non-decreasing functions, hence it is of
  bounded variation. By standard results, see \cite[Cor.~1.4.43]{P},
  it has thus a uniformly converging Fourier series expansion
\[
    F(x) - x \, = \, \sum_{m=1}^{\infty} b_{m} \, \sin(2 \pi m x)\ts .
\]
The Fourier coefficient $b_{m}$ (for $m \in \NN$) is 
\[
\begin{split}
   b_{m} \, & = \, 2 \int_{0}^{1} \sin(2 \pi m x) \,
   \bigl( F(x) - x \bigr) \dd x \: = \: \frac{1}{m\pi}\,
   + 2 \int_{0}^{1} \sin(2 \pi m x) \, F(x) \dd x \\[1mm]
   & = \: \frac{1}{m \pi} \int_{0}^{1} \cos(2 \pi m x) \dd F(x)
   \: = \: \frac{1}{m \pi} \int_{0}^{1} e^{2 \pi i m x} \dd F (x)
   \: = \: \frac{\eta(m)}{m \pi} \ts .
\end{split}
\]
The first step in the second line follows from integration by parts,
while the next is a consequence of the symmetry of $\dd F$ together
with its periodicity (wherefore the imaginary part of the integral
vanishes). Recalling that $F$ (restricted to $[0,1]$) is the
distribution function of the probability measure $\mu$ completes the
argument.
\end{proof}

It is interesting that the autocorrelation coefficients occur as
Fourier coefficients this way. In preparation of a later result, let
us look at this connection more closely. A key observation is that the
recursion relations for $\eta$ from Lemma~\ref{lem:gen-tm-rec} (which
have a unique solution once the initial condition $\eta(0)=1$ is
given, with $\lvert \eta(n)\rvert \le 1$ for all $n\in\ZZ$) can also
be read as a recursion as follows. Let $\beta \in [-1,1]^{\NN}$ be
a sequence and define a mapping $\Psi$ via
\begin{equation}\label{eq:seq-rec}
   \bigl(\Psi \beta\bigr)_{(k+\ell)n + r} :=
   \myfrac{1}{k+\ell} \begin{cases}
     \alpha^{}_{k,\ell,r} +
     \alpha^{}_{k,\ell,k+\ell-r} \, \beta^{}_{n+1} \,  , &
    \text{if $n=0$ and $1 \le r < k+\ell$,} \\[1mm]
    (k+\ell)\,  \beta^{}_{n} \, , & 
    \text{if $n\in \NN$ and $r=0$,} \\[1mm]
    \alpha^{}_{k,\ell,r} \, \beta^{}_{n}
   + \alpha^{}_{k,\ell,k+\ell-r} \, \beta^{}_{n+1} \,  , &
    \text{if $n\in \NN$ and $1 \le r < k+\ell$,}
   \end{cases}
\end{equation}
which completely defines the sequence $\Psi \beta$. This mapping
derives from the recursion for $\eta$ with positive arguments
when $\eta(0)=1$.

\begin{lemma}\label{lem:seq-rec}
  The mapping\/ $\Psi$ maps $[-1,1]^{\NN}$ into itself, with\/ $\|
  \Psi \beta \|^{}_{\infty} \le \| \beta \|^{}_{\infty}$. Moreover,
  for any $\beta^{(0)}\in [-1,1]^{\NN}$, the iteration sequence
  $\beta^{(N)}$ defined by $\beta^{(N+1)} = \Psi \beta^{(N)}$ for
  $N\ge 0$ converges pointwise towards
  $\bigl(\eta(n)\bigr)_{n\in\NN}$.
\end{lemma}
\begin{proof}
The first claim follows from $\lvert \alpha^{}_{k,\ell,r} \rvert +
\lvert \alpha^{}_{k,\ell,k+\ell-r} \rvert \le k+\ell$, which was used
earlier in the proof of Lemma~\ref{lem:growth}, via the triangle
inequality. When $r=0$ or when $k=\ell=r$, one has equality here,
so that $\Psi$ is not a contraction on $[-1,1]^{\NN}$ for the
supremum norm.

Observe that the iteration for $\beta^{}_{1}$ is closed and reads
\[
    \beta^{\ts\prime}_{1} \, = \, \myfrac{k+\ell-3}{k+\ell}
    - \myfrac{1}{k+\ell}\, \beta^{}_{1}\, ,
\]
which is an affine mapping with Lipschitz constant $\frac{1}{k+\ell}$
and hence a contraction. The iteration for $\beta^{}_{1}$ thus
converges exponentially fast (to $\frac{k+\ell-3}{k+\ell+1}$) by
Banach's contraction principle.

What happens with the iteration for $\beta^{}_{1}$ determines
everything else, because the components $\beta_{n}$ with $(k+\ell)^{m}
\le n < (k+\ell)^{m+1}$ and $m \ge 0$ emerge from $\beta^{}_{1}$ in
$m$ steps of the iteration.  In particular, the iteration also closes
on any finite block with $1 \le n < (k+\ell)^{m}$ and fixed $m\in
\NN$, and shows exponentially fast convergence. Note though that the
iteration is only non-expanding as soon as $m > 1$, while $\Psi$
induces an affine mapping with Lipschitz constant
\[
    L \, = \, \frac{\max \big\{ \lvert 
    \alpha^{}_{k,\ell,k+\ell-r} \rvert \,\big|\,
    1 \le r < k+\ell \big\} }{k+\ell} \, \le \,
    \frac{k+\ell-1}{k+\ell} \, < \, 1 \, , 
\]
on the components $\beta_{n}$ with $1\le n < k+\ell$.

Pointwise convergence is now clear, and the limit is the one
specified by the original recursion, which proves the claim.
\end{proof}

The recursion relations for $\eta$ can also be used to derive a
functional equation for the distribution function $F$. Observe first
that
\begin{equation}\label{eq:left}
\begin{split}
  \eta\bigl((k\!+\!\ell)m+r\bigr) \, & = \,
   \int_{0}^{1} e^{2\pi i ((k+\ell)m+r)x} \dd F(x) \,
   =   \int_{0}^{k+\ell} \! e^{2\pi im x} \, e^{2\pi i \frac{rx}{k+\ell}} 
       \dd F\bigl(\tfrac{x}{k+\ell}\bigr) \\[1mm]
  & =  \int_{0}^{1} \! e^{2\pi im x} \, e^{2\pi i \frac{rx}{k+\ell}} 
       \sum_{s=0}^{k+\ell-1} e^{2\pi i \frac{rs}{k+\ell}} 
       \dd F\bigl(\tfrac{x+s}{k+\ell}\bigr) .
\end{split}
\end{equation}
On the other hand, we know from Lemma~\ref{lem:gen-tm-rec} that 
\begin{equation}\label{eq:right}
\begin{split}
  \eta\bigl((k\!+\!\ell)m+r\bigr) \, & = \,
   \myfrac{1}{k+\ell}
     \bigl( \alpha^{}_{k,\ell,r}\, \eta(m) +
     \alpha^{}_{k,\ell,k+\ell-r} \, \eta(m\!+\!1) \bigr)\\[1mm] 
  & = \, 
  \int_{0}^{1} e^{2\pi i m x} \, \frac{\alpha^{}_{k,\ell,r} + 
  \alpha^{}_{k,\ell,k+\ell-r}\, e^{2\pi i x}}{k+\ell}\dd F(x)\, . 
\end{split}
\end{equation}
A comparison of \eqref{eq:left} and \eqref{eq:right} leads to the
following result.

\begin{prop}\label{prop:fe}
The distribution function $F$ for $k,\ell\in\NN$ satisfies the
functional equation
\[
   F(x) \, = \, \myfrac{1}{k+\ell}\int_{0}^{(k+\ell)x}\!\! 
   \vartheta\bigl(\tfrac{y}{k+\ell}\bigr) \dd F(y)
   \quad\text{with}\quad
   \vartheta(x)=1+\myfrac{2}{k+\ell}\sum_{r=1}^{k+\ell-1}\! 
   \alpha^{}_{k,\ell,r} \, \cos(2\pi rx)\, .
\]
This relation holds for all $x\in\RR$, and $\vartheta$ is continuous 
and non-negative.
\end{prop}
\begin{proof}
Equations~\eqref{eq:left} and \eqref{eq:right}, which hold for all
$m\in\ZZ$, state the equality of the Fourier coefficients of two
$1$-periodic Riemann-Stieltjes measures, which must thus be equal (as
measures). For $0\le r<k\!+\!\ell$, we thus have 
\[
   \sum_{s=0}^{k+\ell-1} e^{2\pi i \frac{rs}{k+\ell}} 
       \dd F\bigl(\tfrac{x+s}{k+\ell}\bigr)  \, = \, 
   e^{-2\pi i \frac{rx}{k+\ell}}\, 
  \frac{\alpha^{}_{k,\ell,r} + 
  \alpha^{}_{k,\ell,k+\ell-r}\, e^{2\pi i x}}{k+\ell}\dd F(x)\, .
\]
Fix an integer $t$ with $0\le t<k\!+\!\ell$ and multiply the equation
for $r$ on both sides by $\exp\bigl(-2\pi i \frac{r t}{k+\ell}\bigr)$.
Since $\sum_{r=0}^{k+\ell-1}\exp\bigl(2\pi i \frac{r (s-t)}{k+\ell}
\bigr)= (k \! + \! \ell)\, \delta_{s,t}$, a summation over $r$
followed by a division by $(k+\ell)$ leads to
\[
\begin{split}
  \dd F\bigl(\tfrac{x+t}{k+\ell}\bigr) \, & = \,
   \myfrac{1}{k+\ell}\left( 1 + \sum_{r=1}^{k+\ell-1} \left(
   \frac{\alpha^{}_{k,\ell,r}}{k+\ell} \, e^{-2\pi i
   \frac{r(x+t)}{k+\ell}} + \frac{\alpha^{}_{k,\ell,k+\ell-r}}{k+\ell}
   \, e^{2\pi i \frac{(k+\ell-r)x-rt}{k+\ell}}\right)\right) \dd
   F(x)\\[1mm] & = \, \myfrac{1}{k+\ell}\,\biggl( 1 +
   \myfrac{2}{k+\ell} \sum_{r=1}^{k+\ell-1} \!\alpha^{}_{k,\ell,r} \,
   \cos\Bigl(2\pi \tfrac{r(x+t)}{k+\ell}\Bigr)\biggr) \dd F(x) \, = \,
   \frac{\vartheta\bigl(\frac{x+t}{k+\ell}\bigr)}{k+\ell}\, \dd F(x)
\end{split}   
\]
which is valid for all $x\in[0,1)$.

To derive the functional equation, we need to calculate $F(x)$ and
hence to integrate the previous relations with an appropriate
splitting of the integration region. When $[y]$ and $\{y\}$ denote the
integer and the fractional part of $y$, one finds
\begin{equation}\label{eq:fsplit}
   F(x) \,=  \int_{0}^{x} \dd F(y) 
   \, =  \int_{0}^{\{(k+\ell)x\}} \!
   \dd F\bigl(\tfrac{y+[(k+\ell)x]}{k+\ell}\bigr) \; +  \!\!
   \sum_{t=0}^{[(k+\ell)x]-1} \!\int_{0}^{1} \!
   \dd F\bigl(\tfrac{y+t}{k+\ell}\bigr)\, ,
\end{equation}
which holds for all $x\in [0,1)$. Observe next that 
\[
   \int_{0}^{1} \!
   \dd F\bigl(\tfrac{y+t}{k+\ell}\bigr) \, = \,
   \myfrac{1}{k+\ell} \int_{0}^{1} 
   \vartheta\bigl(\tfrac{y+t}{k+\ell}\bigr) \dd F(y) \, = \, 
   \myfrac{1}{k+\ell} \int_{t}^{t+1} \!
   \vartheta\bigl(\tfrac{z}{k+\ell}\bigr) \dd F(z) \, .
\]
holds for any $0\le t\le k\!+\!\ell$, where we used that $\dd
F(z-n)=\dd F(z)$ for all $n\in \ZZ$ as a consequence of the relation
$F(z+1)=1+F(z)$ for $z\in\RR$.  Similarly, one obtains
\[
   \int_{0}^{\{(k+\ell)x\}} \!
   \dd F\bigl(\tfrac{y+[(k+\ell)x]}{k+\ell}\bigr) \, = \, 
   \myfrac{1}{k+\ell}\int_{[(k+\ell)x]}^{(k+\ell)x} 
   \vartheta\bigl(\tfrac{z}{k+\ell}\bigr) \dd F(z)\, .
\] 
We can now put the pieces in \eqref{eq:fsplit} together to obtain the
functional equation as claimed, which clearly holds for all $x\in \RR$.

The continuity of $\vartheta$ is clear. Its non-negativity follows
from the functional equation, because we know that $F$ is
non-decreasing on $[0,1]$ (as it is the distribution function of the
positive measure $\mu$). If we had $\vartheta(a)<0$ for some $a\in
[0,1]$, there would be some $\varepsilon>0$ such that $\vartheta(y) <
-\varepsilon$ in a neighbourhood of $a$, which would produce a
contradiction to the monotonicity of $F$ via the functional equation.
\end{proof}

\begin{remark}[\textsl{Properties of the integration kernel}]
   The non-negative function $\vartheta$ of Proposition~\ref{prop:fe}
   has various interesting and useful properties. Among them are the
   normalisation relations
\[
    \int_{0}^{1} \vartheta (x) \dd x \, = \, 1 
    \quad \text{and} \quad
    \int_{0}^{\frac{1}{2}} \vartheta (x) \dd x 
    \, = \, \myfrac{1}{2} \ts ,
\]
which follow from $\int_{0}^{1} \cos (2 \pi r x) \dd x = 0$ for $r\ne
0$ together with the $1$-periodicity of $\vartheta$ and its symmetry
(whence we also have $\vartheta(1-x) = \vartheta(x)$).  Another is
the bound
\[
    \| \vartheta \|^{}_{\infty} \, \le \, q \, ,
\]
  where $q$ is the number from Lemma~\ref{lem:growth} (when
  $k+\ell > 2$) or $q=2$ (when $k=\ell=1$). This bound is
  proved by another use of Eq.~\eqref{eq:alsum} and the lines
  following it. \exend
\end{remark}

The functional equation of Proposition~\ref{prop:fe} provides the basis
for the calculation of $F$ by a Volterra-type iteration. To this end,
one defines $F^{}_{0}(x)=x$ (so that $\dd F^{}_{0}(x)=\dd x$) together
with the recursion
\begin{equation}\label{eq:volt}
   F_{n+1}(x) \, =\, \myfrac{1}{k+\ell}\int_{0}^{(k+\ell)x}
   \!\! 
   \vartheta\bigl(\tfrac{y}{k+\ell}\bigr) \dd F_{n}(y)
\end{equation}
for $n\ge 0$. It is clear that each $F_{n}$ defines an absolutely
continuous Riemann-Stieltjes measure, so that one can define
densities $f_{n}$ via $\dd F_{n}(x) = f_{n}(x)\dd x$. This gives
\[
  \int_{0}^{x} f_{n+1}(z) \dd z \, = \, F_{n+1}(x) 
  \, =\, \int_{0}^{x} \vartheta(z) \, 
  f_{n}\bigl((k+\ell)z\bigr) \dd z \, ,
\]
which results in $f_{n+1}(z) = \vartheta(z)\,
f_{n}\bigl((k+\ell)z\bigr)$, and hence in the continuous function
\begin{equation}\label{eq:riesz}
   f_{n}(z) \, = \, \prod_{j=0}^{n-1} 
   \vartheta\bigl((k+\ell)^{j}z\bigr) .
\end{equation}

To put this iteration into perspective, let us introduce the space $D$
of non-decreasing and continuous real-valued functions $G$
on $\RR$ that satisfy $G(-x)=-G(x)$ and $G(x+1)=1+ \ts G(x)$ for all
$x\in\RR$. In particular, this implies $G(q)=q$ for all
$q\in\frac{1}{2}\ZZ$. We equip this space with the
$\|.\|_{\infty}$-norm, and thus with the topology induced by uniform
convergence, in which the space is closed and complete.  Each $G\in D$
defines a positive Riemann-Stieltjes measure on $\RR$ that is
reflection symmetric and $1$-periodic.  Also, $G(x)-x$ always defines
a continuous, skew-symmetric and $1$-periodic function of bounded
variation. Our distribution functions $F$ from above are elements of
$D$.

Let $k,\ell \in \NN$ be fixed. Define a mapping $\Phi$ by
$G\mapsto\Phi G$, where
\begin{equation}\label{eq:Gmap}
   \Phi G(x) \,= \, \myfrac{1}{k+\ell} \int_{0}^{(k+\ell)x} 
   \!\vartheta\bigl(\tfrac{y}{k+\ell}\bigr)\dd G(y) \, .
\end{equation}
Clearly, our previous iteration \eqref{eq:volt} can now be
written as $F_{n+1}=\Phi F_{n}$ with the initial condition
$F^{}_{0}(x)=x$, where $F^{}_{0}\in D$.

\begin{lemma}
  The operator\/ $\Phi$ maps $D$ into itself. Moreover, for arbitrary
  $F^{(0)} \in D$, the iteration sequence defined by $F^{(N+1)} := \Phi
  F^{(N)}$ for $N\ge 0$ converges uniformly to the continuous distribution
  function $F$ of Proposition~$\ref{prop:distribution}$.
\end{lemma}
\begin{proof}
  Let $G\in D$. It is clear that $\Phi G$ is again continuous (since
  $\vartheta$ and $G$ are continuous) and non-decreasing
  (since $\vartheta$ is non-negative by Proposition~\ref{prop:fe}).
  The skew-symmetry follows from $\dd G(-y) = -\dd G(y)$ by a simple
  calculation. Finally, one has
\[
   \Phi G(x+1) \, = \, \myfrac{1}{k+\ell} \int_{0}^{(k+\ell)(x+1)} 
   \vartheta\bigl(\tfrac{y}{k+\ell}\bigr) \dd G(y) \, = \, 
   \Phi G(x) + I\, ,
\]
  where the remaining integral $I$, using $\vartheta$ from
  Proposition~\ref{prop:fe} and the periodicity of $\dd G$, is
\[
\begin{split}
   I \, & = \, \myfrac{1}{k+\ell}\int_{(k+\ell)x}^{(k+\ell)x + (k+\ell)}
   \vartheta\bigl(\tfrac{y}{k+\ell}\bigr) \dd G(y) \: = \: 
   \myfrac{1}{k+\ell}\int_{0}^{k+\ell}
   \vartheta\bigl(\tfrac{y}{k+\ell}\bigr) \dd G(y) \\[1mm]
    & = \, \myfrac{1}{k+\ell} \sum_{t=0}^{k+\ell-1} \int_{t}^{t+1}
    \vartheta\bigl(\tfrac{y}{k+\ell}\bigr) \dd G(y) \:
     = \: \myfrac{1}{k+\ell} \sum_{t=0}^{k+\ell-1} \int_{0}^{1}
    \vartheta\bigl(\tfrac{y+t}{k+\ell}\bigr) \dd G(y) \\[1mm]
    & = \, 1 + \myfrac{2}{(k+\ell)^{2}} \sum_{r=1}^{k+\ell-1}
    \alpha^{}_{k,\ell,r}\int_{0}^{1} \sum_{t=0}^{k+\ell-r}
    \cos\bigl(2\pi\tfrac{ry+rt}{k+\ell}\bigr)\dd G(y)\: = \: 1\, .
\end{split} 
\]
The last step follows because the sum under the integral vanishes as a
consequence of the relation $\sum_{t=0}^{m-1}\exp\bigl(2\pi i
\,\frac{r t}{m}\bigr) = 0$ for $1\le r<m$.

To establish the convergence, let us first show pointwise convergence
$F^{(N)} (x) \xrightarrow{N\to\infty} F(x)$ for all $x\in\RR$.
Observe that $F^{(N)} (x) - x$ defines a $1$-periodic, continuous
function on $\RR$ for each $N\in\NN_{0}$, and so does $F(x) - x$.
Their restrictions to $[0,1]$ define regular (signed) measures on the
unit circle, $\nu^{(N)}$ and $\nu$ say, with Fourier-Stieltjes
coefficients $a^{(N)}_{n}$ and $a^{}_{n}$, where $n\in\ZZ$. The latter
coefficients follow from Proposition~\ref{prop:distribution}, and the
former from the observation that every $G\in D$ has a uniformly
converging Fourier series of the form
\[
   G (x) \, = \, x + \sum_{n=1}^{\infty} 
   \frac{\beta_{n}}{n \ts \pi} \sin (2 \pi n x) \ts ,
\]
where $\lvert \beta_{n} \rvert \le 1$ by an application of
\cite[Thm.~II.4.12]{Z}. Here, we have $\beta^{(N+1)} = \Psi
\beta^{(N)}$ with the mapping $\Psi$ of Lemma~\ref{lem:seq-rec}, with
initial condition $\beta^{(0)}$ and convergence $\beta^{(N)}_{n}
\xrightarrow{N\to\infty}\eta(n)$ for each $n\in\NN$. Consequently,
$a^{(N)}_{n} \xrightarrow{N\to\infty} a^{}_{n}$ for all $n\in\ZZ$,
which means that $\nu^{(N)} (p) \xrightarrow{N\to\infty} \nu (p)$ for
any trigonometric polynomial $p$ with period $1$, and hence (by an
application of the Stone-Weierstrass theorem \cite{RS}) for all
continuous functions on $[0,1]$. This proves weak convergence
$\nu^{(N)} \xrightarrow{N\to\infty} \nu$ on the unit circle. Since all
measures are absolutely continuous, with continuous Radon-Nikodym
densities, this implies pointwise convergence $F^{(N)} (x) - x
\xrightarrow{N\to\infty} F(x) - x$ for all $x\in [0,1]$, and hence (by
periodicity) the pointwise convergence on $\RR$ claimed above.

Uniform convergence on $\RR$ now follows from that on $[0,1]$, which
can be shown via the `stepping-stone' argument from the proof of
\cite[Thm.~30.13]{Bauer}. For completeness, we spell out the
details. Given $\varepsilon > 0$, there are numbers $m\in\NN$ and
$0=x^{}_{0} < x^{}_{1} < \cdots < x^{}_{m}=1$ with
\[
    \lvert F( x_{i} ) - F( x_{i-1}) \rvert \, = \,
     F(x_{i}) - F(x_{i-1}) \, < \, \varepsilon
\]
for $1 \le i \le m$. Also, for all sufficiently large $N\in\NN$, one
has
\[
   \lvert F^{(N)} (x_{i}) - F (x_{i}) \rvert \, < \, \varepsilon
\]
for all $0 \le i \le m$. Since $F^{(N)} (1) = F(1) = 1$, consider now an 
arbitrary $x\in [0,1)$, so that $x_{i-1} \le x < x_{i}$ for precisely one 
$i\in \{1,2,\ldots,m\}$. Using monotonicity, this implies the inequalities
\[
    F (x_{i-1}) \, \le \, F (x) \, \le \, F (x_{i})
    \, < \, F (x_{i-1}) + \varepsilon
\]
and
\[
    F (x_{i-1}) - \varepsilon \, < \, F^{(N)} (x_{i-1})
    \, \le \, F^{(N)} (x) \, \le \, F^{(N)} (x_{i})
    \, < \, F (x_{i}) + \varepsilon \, < \,
    F (x_{i-1}) + 2 \varepsilon \ts .
\]
Together, they give $\lvert F^{(N)} (x) - F (x) \rvert <
2 \varepsilon$, which holds for all $x\in [0,1]$, and then
uniformly for all $x\in\RR$, as $F^{(N)} - F$ is $1$-periodic
for all $N\in\NN_{0}$.
\end{proof}

Due to our convergence results, Equation~\eqref{eq:riesz} means that
the measure $\widehat{\gamma}$ has a (vaguely convergent)
representation as the infinite \emph{Riesz product} $ \prod_{n\ge 0}
\vartheta \bigl( (k+\ell)^{n} x \bigr)$.  The entire analysis is thus
completely analogous to that of the original TM sequence and shows
that the latter is a typical example in an infinite family. Two further
cases are illustrated in Figure~\ref{fig2}.

\begin{figure}
\bigskip
\begin{center}
\!\!\!\!\!\!
\includegraphics[width=0.48\textwidth]{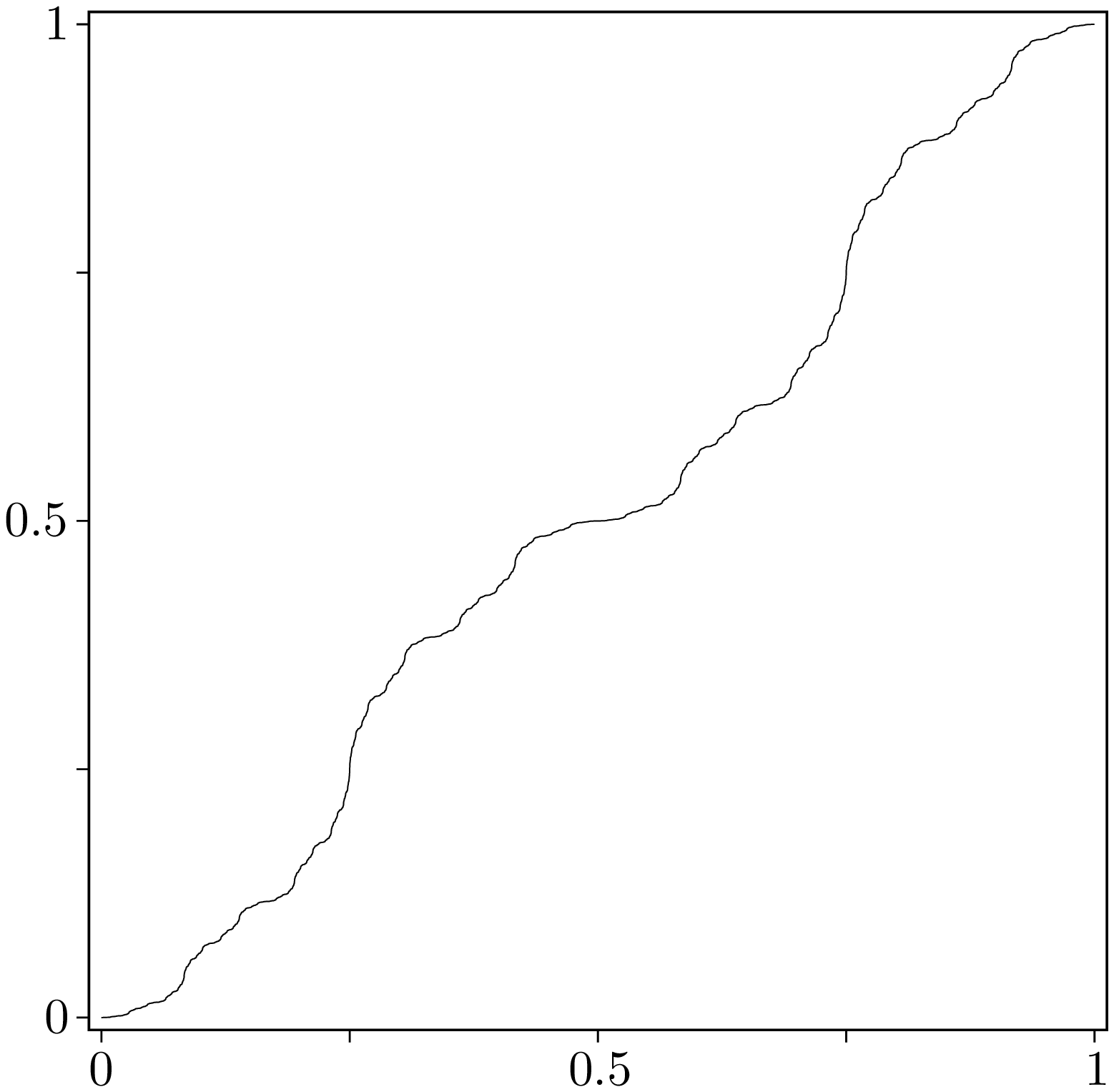} \quad
\includegraphics[width=0.48\textwidth]{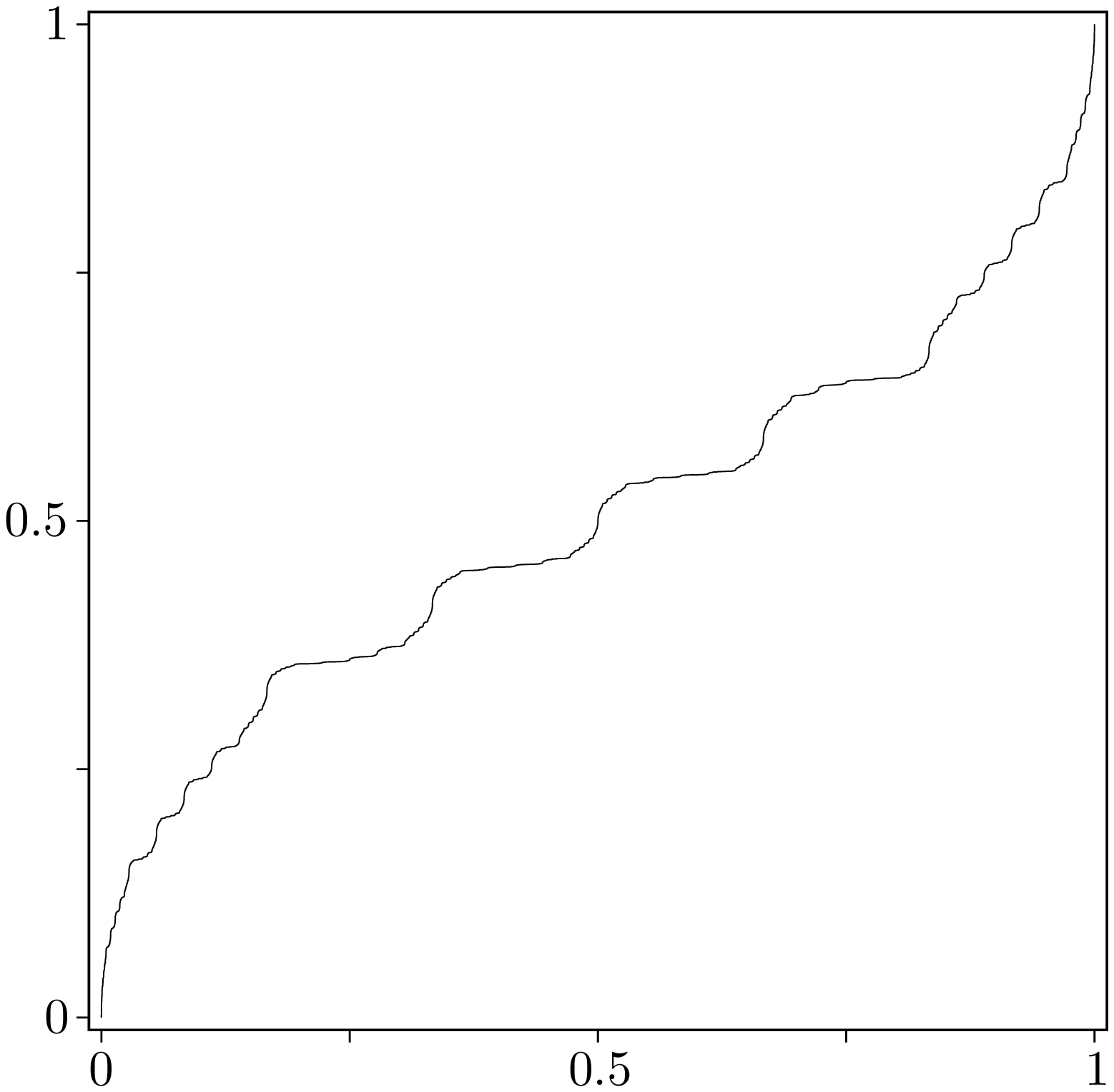}
\end{center}
\caption{\label{fig2}The continuous and strictly increasing
  distribution functions of the generalised Thue-Morse measures on $[0,1]$
  for parameters $(k,\ell)=(2,1)$ (left) and $(5,1)$ (right).}
\end{figure}

\begin{remark}[\textsl{Pure point factors}]\label{rem:Kronecker}
The block map \eqref{eq:block-map} applies to any member of the
gTM family, and always gives a 2-to-1 cover of the hull
$\XX^{\mathrm{pd}}_{k,\ell}$ of the \emph{generalised period doubling} (gpd)
substitution
\begin{equation}\label{eq:defgpd}
   \varrho^{\ts\prime} =  \varrho^{\ts\prime}_{k,\ell} \! : \quad
   a \mapsto b^{k-1} a b^{\ell-1} b \ts , \quad
   b \mapsto b^{k-1} a b^{\ell-1} a \ts .
\end{equation}
Since we always have a coincidence (at the $k$th position) in the
sense of Dekking \cite{D}, they all define systems with pure point
spectrum (which can be described as model sets in the spirit of
\cite{BMS,BM}) -- another analogy to the classic case $k=\ell=1$.
Also, for given $k,\ell \in \NN$, the dynamical system
$(\XX^{\mathrm{pd}}_{k,\ell},\ZZ)$ is a topological factor of
$(\XX^{\mathrm{TM}}_{k,\ell},\ZZ)$.

Moreover, the dynamical spectrum of the gTM system contains $\ZZ
[\frac{1}{k+\ell}]$ as its pure point part, which happens to be the
entire spectrum of the gpd system. The latter fact can be derived from
the support of the gpd diffraction measure, via the general
correspondence between the dynamical and the diffraction spectrum for
pure point systems \cite{LMS,BL}. The detailed calculations can be
done in analogy to the treatment of the period doubling system in
\cite{BM,BG11}, after an explicit formulation of the one-sided fixed
point of the gpd substitution, which results in the Fourier module
$\ZZ [\frac{1}{k+\ell}]$. Consequently, the gpd system is a
topological factor of the gTM system with maximal pure point spectrum.  

The gpd system can be described as a model set (with suitable $p$-adic
type internal space). From \cite{BLM}, we know that there exists an
almost everywhere $1$-to-$1$ `torus' parametrisation via a solenoid
$\mathbb{S}_{k+\ell}$.  Here, the inflation acts as multiplication by
$k \! + \!\ell$. In fact, the solenoid provides the maximum
equicontinuous (or Kronecker) factor of the gTM system.  \exend
\end{remark}

\section{Topological invariants}\label{sec:topol}

The formulation via Dirac combs embeds the symbolic sequences into the
class of translation bounded measures on $\RR$, as described in a
general setting in \cite{BL}. It is thus natural to (also) consider
the continuous counterpart of the discrete hull $\XX$ in the form
\[
    \YY \, = \, \overline{ \{\delta_{t} * \omega \mid t\in \RR \} } \ts ,
\]
where $\omega$ is an arbitrary element of $\XX$ (which is always
minimal in our situation) and the closure is now taken in the vague
topology. Here, $\YY$ is compact, and $(\YY,\RR)$ is a topological
dynamical system. Note that the continuous tiling hull mentioned
earlier is topologically conjugate, wherefore we use the same symbol
for both versions. The discrete hull $\XX$, in this formulation, is
homeomorphic with the compact set
\[
    \YY^{}_{0} \, = \, \{ \nu \in \YY \mid \nu (\{0\}) \ne 0\}
            \, \subset \, \YY \ts ,
\]
while $(\YY,\RR)$ is the suspension of the system $(\YY^{}_{0},\ZZ)$.

We are now going to construct the continuous hulls of the generalised
Thue-Morse and period doubling sequences as inverse limits of the
inflation map acting on an AP-complex \cite{AP}, and use this construction
to compute their \v{C}ech cohomology. At first sight, there are infinitely
many cases to be considered. We note, however, that the AP complex 
$\Gamma$ depends only on the atlas of all $r$-patterns occurring in the 
tiling, for some bounded $r$. In the original AP construction \cite{AP},
$\Gamma$ consists of $3$-patterns -- tiles with one collar tile on the 
left and one on the right. How these collared tiles are glued together in 
the complex is then completely determined by the set of $4$-patterns 
occurring in the tiling. Therefore, there are only finitely many different 
AP complexes $\Gamma$ to be considered.

For obvious reasons, we want to extend the factor map $\phi\!:\,
\YY^{\mathrm{TM}}_{k,\ell}\rightarrow\YY^{\mathrm{pd}}_{k,\ell}$
between the continuous hulls to the cell complexes approximating them,
compare (\ref{eq:diagram-approx-tm}). We therefore choose the cell
complex $\Gamma^{\mathrm{TM}}_{k,\ell}$ to consist of tiles with one
extra layer of collar on the right, compared to the cell complex
$\Gamma^{\mathrm{pd}}_{k,\ell}$, so that we obtain, for each pair
$(k,\ell)$, a well-defined factor map
$\phi\!:\,\Gamma^{\mathrm{TM}}_{k,\ell} \rightarrow
\Gamma^{\mathrm{pd}}_{k,\ell}$, which we denote by the same symbol.

While the extra collar seems to complicate things, we can compensate
this by a simplification compared to the original AP construction.  In
\cite{GM}, it was shown that a computation using a complex with
one-sided collars already yields the correct cohomology groups, even
though the inverse limit is not necessarily homeomorphic to the
continuous hull of the tiling. The minimal setup therefore consists of
AP complexes $\Gamma^{\mathrm{TM}}_{k,\ell}$ with tiles having
two-sided collars, and complexes $\Gamma^{\mathrm{pd}}_{k,\ell}$
having tiles with left collars only. All collars have thickness one.
A straight-forward analysis shows that only three cases have to be
distinguished: $k=\ell=1$ (the classical case), either $k=1$ or
$\ell=1$ (but not both), and $k,\ell\ge2$. The complexes
$\Gamma^{\mathrm{TM}}_{k,\ell}$ are all sub-complexes of the complex
shown in Figure~\ref{fig:gentm}.  If $k=\ell=1$, the two loops on the
left and right have to be omitted, because the corresponding patterns
$111$ and $\bar{1}\bar{1}\bar{1}$ do not occur in the classic TM
sequence. Likewise, in the case $k,\ell\ge2$, the lens in the centre
has to be omitted, whereas the full complex has to be used in all
remaining cases.

\begin{figure}
\begin{center}
\includegraphics[width=0.8\textwidth]{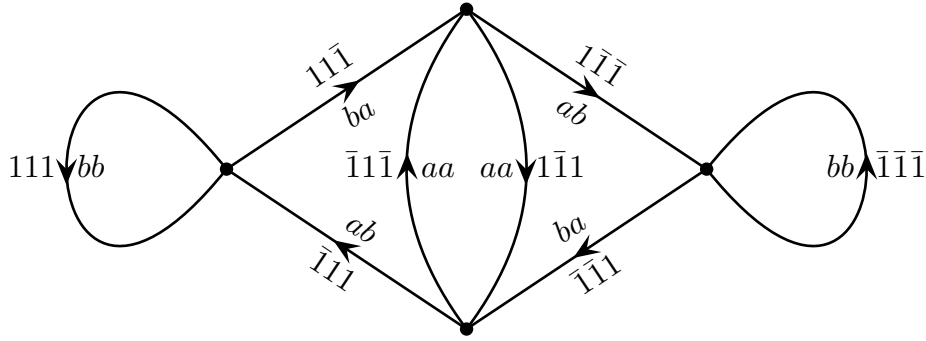} 
\end{center}
\caption{\label{fig:gentm}The universal complex for the gTM system.
Edges are labelled both with the corresponding collared tile of the 
gTM tiling and its image in the gpd tiling.}
\end{figure}

For the gpd sequences, an approximant complex that includes all different 
cases as sub-complexes can be constructed in a similar way. This complex
is shown in Figure~\ref{fig:genpd}. As we have labelled the edges in 
Figure~\ref{fig:gentm} with both the 3-pattern in the gTM sequence 
and the corresponding image under $\phi$ in the gpd sequence, it is 
easy to see how the factor map $\phi$ maps the complex of 
Figure~\ref{fig:gentm} to that of Figure~\ref{fig:genpd}: the loops on 
the left and right of Figure~\ref{fig:gentm} are both mapped onto the 
right loop of Figure~\ref{fig:genpd}, the central lens in 
Figure~\ref{fig:gentm} is wrapped twice around the left loop in 
Figure~\ref{fig:genpd}, and the rhombus in Figure~\ref{fig:gentm} is 
wrapped twice around the central loop in Figure~\ref{fig:genpd}. 
Therefore, $\Gamma^{\mathrm{pd}}_{k,\ell}$ consists of the left and the 
central loop of  Figure~\ref{fig:genpd} in the case $k=\ell=1$ (the 
classical case), the central and the right loop of Figure~\ref{fig:genpd} 
in the case $k,\ell\ge2$, and of all three loops in the remaining cases.
It is obvious that the map $\phi$ is uniformly $2$-to-$1$ also at the
level of the approximant complexes.

\begin{figure}
\begin{center}
\includegraphics[width=0.5\textwidth]{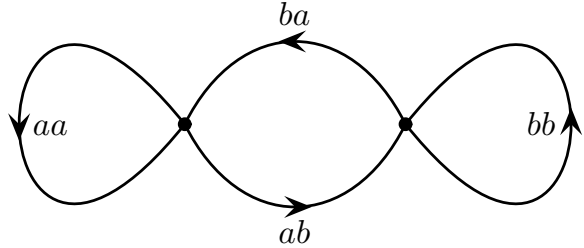} 
\end{center}
\caption{\label{fig:genpd}The universal complex for the
          gpd system. Its edges correspond to tiles with one-sided
          (left) collar.}
\end{figure}

The cohomology of the continuous hulls $\YY^{\mathrm{TM}}_{k,\ell}$ and
$\YY^{\mathrm{pd}}_{k,\ell}$ is now given by the direct limit of the induced
maps $\rho^*_{k,\ell}$ and  $\rho'^*_{k,\ell}$ on the cohomology of the
AP complexes, $H^*(\Gamma^{\mathrm{TM}}_{k,\ell})$ and 
$H^*(\Gamma^{\mathrm{pd}}_{k,\ell})$. While there are only 3 different cell 
complexes in each case, there is an infinite family of maps  
$\rho^*_{k,\ell}$, respectively  $\rho'^*_{k,\ell}$, but these can be 
parametrised by $k$ and $\ell$. 

We first define a basis of the homology of the full complexes of
Figures~\ref{fig:gentm} and \ref{fig:genpd}, subsets of which shall be
used for all $k$ and $\ell$, and pass later to the corresponding dual
basis for the cohomology. The homology of the full complex
$\Gamma^{\mathrm{TM}}_{k,\ell}$ is generated by the basis
\[
c_1^\mathrm{TM} = 1\bar{1}1, \quad c_2^\mathrm{TM} = \bar{1}1\bar{1}, \quad
c_3^\mathrm{TM} = 1\bar{1}\bar{1}1, \quad c_4^\mathrm{TM} = 1, \quad
c_5^\mathrm{TM} = \bar{1},
\] 
where the words on the right have to be thought of as being repeated
indefinitely. Likewise, the homology of
$\Gamma^{\mathrm{pd}}_{k,\ell}$ is generated by the cycles
\[
    c_1^\mathrm{pd} =\, a, \quad c_2^\mathrm{pd} =\, ab,
     \quad c_3^\mathrm{pd} =\, b.
\]
Clearly, we have
\[
\begin{gathered}
\phi_*(c_1^\mathrm{TM})\, =\, c_1^\mathrm{pd} + c_2^\mathrm{pd},\quad
\phi_*(c_2^\mathrm{TM})\, =\, c_1^\mathrm{pd} + c_2^\mathrm{pd},\quad
\phi_*(c_3^\mathrm{TM})\, =\, 2c_2^\mathrm{pd},\\
\phi_*(c_4^\mathrm{TM})\, =\, c_3^\mathrm{pd},\quad
\phi_*(c_5^\mathrm{TM})\, =\, c_3^\mathrm{pd}.
\end{gathered}
\]
Here, the lower asterisk in $\phi_*$ denotes the induced action on
homology, whereas $\phi^*$ denotes the action on cohomology.

In order to go into more detail, we have to distinguish
the three different cases. We begin with the classical case,
$k=\ell=1$. Here, the relevant basis elements are $c_1^\mathrm{TM}$,
$c_2^\mathrm{TM}$ and $c_3^\mathrm{TM}$ for TM, and $c_1^\mathrm{pd}$ and
$c_2^\mathrm{pd}$ for pd. On these, the substitution acts as 
\[
\begin{gathered}
\rho_*(c_1^\mathrm{TM})\, =\, c_1^\mathrm{TM} + c_2^\mathrm{TM},\quad
\rho_*(c_2^\mathrm{TM})\, =\, c_1^\mathrm{TM} + c_2^\mathrm{TM},\quad
\rho_*(c_3^\mathrm{TM})\, =\, 2\ts (c_1^\mathrm{TM} + c_2^\mathrm{TM})
    - c_3^\mathrm{TM},\\
\rho'_*(c_1^\mathrm{pd})\, =\, c_2^\mathrm{pd},\quad
\rho'_*(c_2^\mathrm{pd})\, =\, 2c_1^\mathrm{pd} + c_2^\mathrm{pd},
\end{gathered}
\]
where $\rho_*$ and $\rho'_*$ again denote the induced action on homology.
If we express these maps as matrices $A^\mathrm{TM}$ and $A^\mathrm{pd}$
acting from the left on column vectors, with respect to the basis
above, and likewise define a matrix $P$ for the action of $\phi_*$, 
we obtain
\begin{equation}
   A^\mathrm{pd}=\left(\begin{matrix} 0 & 2 \\ 1 & 1 \end{matrix}\right),\quad
   A^\mathrm{TM}=\left(\begin{matrix} 1 & 1 & 2 \\ 1 & 1 & 2 \\
                                0 & 0 & -1\end{matrix}\right),\quad
   P = \left(\begin{matrix} 1 & 1 & 0 \\ 1 & 1 & 2 \end{matrix}\right).
\label{eq:mats-classical}
\end{equation}
The corresponding action on the dual basis of cohomology is simply
given by the transposed matrices, or by the same matrices regarded
as acting from the right on row vectors. We adopt the latter viewpoint
here. It is easy to verify that the matrices (\ref{eq:mats-classical})
satisfy $A^\mathrm{pd}P = PA^\mathrm{TM}$. In other words, the substitution 
action commutes with that of $\phi$. $A^\mathrm{pd}$ and $A^\mathrm{TM}$ have
left eigenvectors and eigenvalues 
\begin{equation}
\left(\begin{matrix} 
   1 &  2 & | &  2 \\
   1 & -1 & | & -1 
\end{matrix}\right) 
\quad\text{and}\quad
\left(\begin{matrix} 
   1 &  1 & 0 & | &  2 \\ 
   0 &  0 & 1 & | & -1 \\
   1 & -1 & 0 & | &  0 
\end{matrix}\right).
\end{equation}

Next, we look at the case $k,\ell\ge 2$, which is still relatively simple.
Here, we have to work with the basis elements $c_3^\mathrm{TM}$, $c_4^\mathrm{TM}$
and $c_5^\mathrm{TM}$ for gTM, and $c_2^\mathrm{pd}$ and $c_3^\mathrm{pd}$ for gpd. 
The action of the substitution is then given by
\[
\begin{aligned}
  \rho_*(c_3^\mathrm{TM}) & \, = \, 3c_3^\mathrm{TM} + 
       (2(k+\ell)-6)(c_4^\mathrm{TM}+c_5^\mathrm{TM}),\\
  \rho_*(c_4^\mathrm{TM}) & \, = \, c_3^\mathrm{TM} + 
       (k-2)c_4^\mathrm{TM} + (\ell-2)c_5^\mathrm{TM},\\
  \rho_*(c_5^\mathrm{TM}) & \, = \, c_3^\mathrm{TM} + 
       (\ell-2)c_4^\mathrm{TM} + (k-2)c_5^\mathrm{TM},\\
  \rho'_*(c_2^\mathrm{pd}) & \, = \, 3c_2^\mathrm{pd} + 
                 2\ts (k+\ell-3)c_3^\mathrm{pd},\\
  \rho'_*(c_3^\mathrm{pd}) & \, = \, 2c_2^\mathrm{pd} + 
          (k+\ell-4)c_3^\mathrm{pd}.
\end{aligned}
\]
The corresponding matrices $A^\mathrm{pd}_{k,\ell}$, $A^\mathrm{TM}_{k,\ell}$ 
and $P^{}_{k,\ell}$ read 
\[
  A^\mathrm{pd}_{k,\ell} =
  \left(\begin{matrix} 3 & 2 \\
      2(k \! + \! \ell \! - \! 3) &
      k \! + \! \ell \! - \! 4 
      \end{matrix}\right),\
  A^\mathrm{TM}_{k,\ell}=\left(\begin{matrix} 
    3      & 1      & 1 \\ 
    2(k \! + \! \ell) \! - \! 6 & k \! - \! 2 & \ell \! - \! 2 \\ 
    2(k \! + \! \ell) \! - \! 6 & \ell \! - \! 2 & k \! - \! 2 
             \end{matrix}\right),\
  P^{}_{k,\ell} = \left(\begin{matrix} 
        2 & 0 & 0 \\ 0 & 1 & 1 \end{matrix}\right),
\]
and satisfy $A^\mathrm{pd}_{k,\ell}P^{}_{k,\ell} =
P^{}_{k,\ell}A^\mathrm{TM}_{k,\ell}$.  The matrices
$A^{\mathrm{pd}}_{k,\ell}$ and $A^{\mathrm{TM}}_{k,\ell}$ have left
eigenvectors and eigenvalues
\begin{equation}
\left(\begin{matrix} 
   2        &  1 & | & k \! + \! \ell \\
   k+\ell-3 & -2 & | &  -1 
\end{matrix}\right) 
\quad\text{and}\quad
\left(\begin{matrix} 
   4       &  1 &  1 &  | & k \! + \! \ell \\ 
  k+\ell-3 & -1 & -1 &  | & -1     \\
   0       &  1 & -1 &  | & k \! - \! \ell 
\end{matrix}\right).
\end{equation}

Finally, the case where $\min(k,\ell)=1$, but $k+\ell>2$, requires the full
AP complex. For notational ease, we consider the case $k=1$, $\ell\ge2$. The
case $\ell=1$, $k\ge2$ is completely analogous. The substitution then acts as 
follows on our basis:
\[
\begin{aligned}
  \rho_*(c_1^\mathrm{TM}) & \, = \, c_2^\mathrm{TM} + c_3^\mathrm{TM} + 
                        (\ell-1)c_4^\mathrm{TM} + (2\ell-3)c_5^\mathrm{TM},\\
  \rho_*(c_2^\mathrm{TM}) & \, = \, c_1^\mathrm{TM} + c_3^\mathrm{TM} + 
                        (\ell-1)c_5^\mathrm{TM} + (2\ell-3)c_4^\mathrm{TM},\\
  \rho_*(c_3^\mathrm{TM}) & \, = \, c_1^\mathrm{TM} + c_2^\mathrm{TM} + 
      c_3^\mathrm{TM} + (2\ell-3)(c_4^\mathrm{TM}+c_5^\mathrm{TM}),\\
  \rho_*(c_4^\mathrm{TM}) & \, = \, c_2^\mathrm{TM} + (\ell-2)c_5^\mathrm{TM},\\
  \rho_*(c_5^\mathrm{TM}) & \, = \, c_1^\mathrm{TM} + (\ell-2)c_4^\mathrm{TM},\\
  \rho'_*(c_1^\mathrm{pd}) & \, = \, c_2^\mathrm{pd} + (\ell-1)c_3^\mathrm{pd},\\
  \rho'_*(c_2^\mathrm{pd}) & \, = \, c_1^\mathrm{pd} + 2c_2^\mathrm{pd} + 
                         (2\ell-3)c_3^\mathrm{pd},\\
  \rho'_*(c_3^\mathrm{pd}) & \, = \, c_1^\mathrm{pd} + c_2^\mathrm{pd} + 
                 (\ell-2)c_3^\mathrm{pd}.
\end{aligned}
\]
The corresponding matrices are
\[
\begin{gathered}
A^\mathrm{pd}_{k,\ell}=\left(\begin{matrix} 
     0   &    1    &   1 \\
     1   &    2    &   1 \\
  \ell-1 & 2\ell-3 & \ell-2  \end{matrix}\right),\quad
A^\mathrm{TM}_{k,\ell}=\left(\begin{matrix} 
      0   &     1   &     1   &    0   &    1   \\
      1   &     0   &     1   &    1   &    0   \\
      1   &     1   &     1   &    0   &    0   \\
   \ell \! - \! 1 & 2\ell \! - \! 3 & 2\ell\! - \! 3 
                  &    0   & \ell \! - \! 2 \\ 
  2\ell \! - \! 3 &  \ell \! - \! 1 & 2\ell \! - \! 3 
                  & \ell \! - \! 2 &    0   \end{matrix}\right),\\
P^{}_{k,\ell} = \left(\begin{matrix} 1 & 1 & 0 & 0 & 0 \\ 
                         1 & 1 & 2 & 0 & 0 \\ 
                         0 & 0 & 0 & 1 & 1  \end{matrix}\right),
\end{gathered}
\]
In this case, $A^\mathrm{TM}_{k,\ell}$ and $A^\mathrm{pg}_{k,\ell}$ 
have left eigenvectors and eigenvalues
\begin{equation}
\left(\begin{matrix} 
    1   &     2           &   1 & | & \ell \! + \! 1  \\
  \ell  &  \ell \! - \! 2 &  -2 & | & -1      \\
    1   &  1 \! - \! \ell &   1 & | &  0  \end{matrix}\right) 
   \quad\text{and}\quad
\left(\begin{matrix} 
      3   &     3   &     4   &    1   &    1 & | & \ell+1  \\
   1 \! - \! \ell &  1 \! - \! \ell &  2 \! - \! \ell 
                  &    1   &    1 & | & -1  \\
      1   &    -1   &     0   &    1   &   -1 & | & 1-\ell  \\
   2 \! - \! \ell &  2 \! - \! \ell & 2 \! - \! 2\ell 
                  &    1   &    1 & | & 0 \\ 
   2 \! - \! \ell &  \ell \! - \! 2 &     0   
                  &    1   &   -1 & | & 0  \end{matrix}\right).
\end{equation}

As the cohomology of the continuous hull is given by the direct limit
of the action of the substitution on the cohomology of the AP
approximant complexes, it is determined by the non-zero eigenvalues of
the matrices $A^\mathrm{TM}_{k,\ell}$ and
$A^\mathrm{pd}_{k,\ell}$. For all $k$ and $\ell$, these eigenvalues
are $k+\ell$ and $-1$ for the gpd case, and $k+\ell$, $-1$ and
$k-\ell$ in the gTM case, where the last eigenvalue is relevant only
if $k\ne\ell$. Since $H^0(\Gamma)=\ZZ$ if $\Gamma$ is connected, and
the substitution action on $H^0(\Gamma)$ is trivial, with eigenvalue
1, we can summarise our cohomology results as follows.

\begin{theorem}
  The \v{C}ech cohomology of the continuous hull of the gTM sequences
  is given by $  H^0(\YY^\mathrm{TM}_{k,\ell}) = \ZZ$ and
\[
  H^1(\YY^\mathrm{TM}_{k,\ell}) \, = \, \begin{cases}
     \ZZ[\frac1{k+\ell}] \oplus \ZZ \oplus \ZZ[\frac1{|k-\ell|}] , &
        \text{ if } |k-\ell|\ge 2 , \\
     \ZZ[\frac1{k+\ell}] \oplus \ZZ^2 , & 
        \text{ if } |k-\ell|=1 , \\
     \ZZ[\frac1{k+\ell}] \oplus \ZZ ,  & 
        \text{ if } k=\ell . \end{cases}
\]
The \v{C}ech cohomology of the continuous hull of the generalised 
period doubling sequences is given by
\[
  H^0(\YY^\mathrm{pd}_{k,\ell}) \, = \, \ZZ,\quad
  H^1(\YY^\mathrm{pd}_{k,\ell}) \, = \, 
  \ZZ[{\textstyle\frac{1}{k+\ell}}] \oplus \ZZ ,
\]
valid for any pair $k,\ell\in\NN$. \qed
\end{theorem}

Since the cohomology of the hulls is given by the direct limit of
the columns of the diagram  
\begin{equation} \label{eq:diagram-approx-gtm}
   \begin{CD}
    H^k(\Gamma^{\mathrm{sol}}_{k,\ell})   @>\psi^*>> H^k(\Gamma^{\mathrm{pd}}_{k,\ell}) 
    @>\phi^*>> H^k(\Gamma^{\mathrm{TM}}_{k,\ell}) \\
    @V \times (k+\ell) VV @V\varrho^{\ts\prime*}VV @VV \varrho^* V \\
    H^k(\Gamma^{\mathrm{sol}}_{k,\ell})   @>\psi^*>> H^k(\Gamma^{\mathrm{pd}}_{k,\ell}) 
    @>\phi^*>> H^k(\Gamma^{\mathrm{TM}}_{k,\ell}) \\
   \end{CD}
\end{equation}
(compare (\ref{eq:diagram-approx-tm})), we now have access also to the
homomorphisms in the sequence
\begin{equation}
  \begin{CD}
    H^*(\mathbb{S}_{k+\ell}) @>\psi^*>> H^*(\YY^{\mathrm{pd}}_{k,\ell})
    @>\phi^*>> H^*(\YY^{\mathrm{TM}}_{k,\ell}).
  \end{CD}
\end{equation}
For all three spaces, $H^0=\ZZ$, and the maps $\phi^*$ and $\psi^*$
acting on them are isomorphisms. Further, it is easy to see that
$\psi^*$ embeds $H^1(\mathbb{S}_{k+\ell})=\ZZ[\frac1{k+\ell}]$
non-divisibly in $H^1(\YY^{\mathrm{pd}}_{k,\ell})$; it is simply mapped
isomorphically onto the summand $\ZZ[\frac1{k+\ell}]$ in
$H^1(\YY^{\mathrm{pd}}_{k,\ell})$.  The same cannot be said of the second
map, however. For all pairs $k,\ell$, the matrix $P$ maps the
eigenvector of $A^{\mathrm{pd}}_{k,\ell}$ with eigenvalue $-1$ to twice
the corresponding eigenvector of $A^{\mathrm{TM}}_{k,\ell}$.  As a
result, the quotient $H^1(\YY^{\mathrm{TM}}_{k,\ell}) / \rho^*(
H^1(\YY^{\mathrm{pd}}_{k,\ell}))$ has a torsion component $\ZZ_2$; the
summand $\ZZ$ of $H^1(\YY^{\mathrm{pd}}_{k,\ell})$ coming from the
eigenvalue $-1$ is mapped to $2\ZZ$ in
$H^1(\YY^{\mathrm{TM}}_{k,\ell})$. This result was already known for the
classic period doubling and Thue-Morse sequences \cite{BS}, and
extends to the generalised sequences. Therefore, in the
case $k=\ell$, where the cohomologies of gTM and gpd are isomorphic as
groups, one should rather regard $H^1(\YY^{\mathrm{pd}}_{k,\ell})$ as
subgroup of index 2 of $H^1(\YY^{\mathrm{TM}}_{k,\ell})$.

\section{Dynamical zeta functions}\label{sec:zeta}

The continuous hull permits to employ the Anderson-Putnam method
\cite{AP} for the calculation of the dynamical zeta function of the
\emph{inflation} action on $\YY$.  The dynamical zeta function
\cite{Ruelle} of a substitution can be viewed as a generating function
for the number of fixed points $a (n)$ under $n$-fold substitution via
\begin{equation}\label{eq:zetadef}
    \zeta (z) \, = \, \exp \Bigl(
    \sum_{n=1}^{\infty} \frac{a (n)}{n}\, z^{n} \Bigr).
\end{equation}
Knowing the fixed point counts $a(n)$, one can calculate the number
$c(n)$ of cycles of length $n$ from the formula
\[
    c(n)\, =\, \frac{1}{n} \sum_{d|n} \mu(\tfrac{n}{d}) \, a(d) \ts , 
\]  
where $\mu$ is the M\"{o}bius function from elementary number theory
(and should not be confused with the measure $\mu$ that appeared
above).

Anderson and Putnam \cite{AP} showed how the dynamical zeta
function of a substitution tiling can be expressed by the action
of the substitution on the cohomology of the AP complex $\Gamma$.
In the one-dimensional case, the zeta function is given by
\[
   \zeta(z) \, = \,
   \frac{\det({\mathbbm{1}}-zA^0)}{\det({\mathbbm{1}}-zA^1)},
\]
where $A^k$ is the matrix of the substitution action on $H^k(\Gamma)$.
In our case, $A^0$ is a $1\times1$ unit matrix, and $A^1$ is diagonalisable,
so that we can rewrite the zeta function as
\[
   \zeta(z) \, = \, \frac{1-z}{\prod_i (1-z\lambda_i)},
\]
where $\lambda_i$ are the eigenvalues of $A^1$. For the gTM and gpd
tilings, these eigenvalues have been derived above, so that we arrive 
(after a simple calculation) at the following theorem.

\begin{theorem}
   Let\/ $k,\ell\in\NN$. The generalised Thue-Morse sequence defined by the
   inflation\/ $\varrho^{}_{k,\ell}$ of\/ \eqref{eq:gen-tm-def}
   possesses the dynamical zeta function
\[
   \zeta^{\mathrm{TM}}_{k,\ell} (z) \, = \,
   \frac{1-z}{(1+z)(1-(k+\ell)z) (1-(k-\ell)z)} \ts ,
\]
   while the induced generalised period doubling sequence, as defined
   by $\varrho^{\ts\prime}_{k,\ell}$ of\/ \eqref{eq:defgpd}, leads to
\[
   \zeta^{\mathrm{pd}}_{k,\ell} (z) \, = \,
   \frac{1-z}{(1+z)(1-(k+\ell)z)} \ts .
\]
In particular, when\/ $k=\ell$, one has\/ $\zeta_{\phantom{l}}^{\mathrm{TM}} =
  \zeta_{\phantom{l}}^{\mathrm{pd}}$.  \qed
\end{theorem}

For our two systems, the corresponding fixed point counts, for
$n\in\NN$, can now be calculated from Eq.~\eqref{eq:zetadef} to be
\begin{equation} \label{eq:fix-count}
    a^{\mathrm{pd}}_{k,\ell} (n) \, = \,
    (k+\ell)^n - \bigl(1 - (-1)^n\bigr)
    \quad \text{and} \quad
    a^{\mathrm{TM}}_{k,\ell} (n) \, = \,
    a^{\mathrm{pd}}_{k,\ell} (n) + (k-\ell)^n .
\end{equation}
It is an interesting exercise to relate the orbits of the two systems
according to the action of the mapping $\phi$ in agreement with these
counts.

There is an interesting connection between the inflation on
$\YY^{\mathrm{TM}}_{k,\ell}$ or $\YY^{\mathrm{pd}}_{k,\ell}$ and the 
multiplication by $(k+\ell)$ on the (matching) solenoid 
$\mathbb{S}_{k+\ell}$, which emerges from the torus parametrisation 
\cite{BHP,M,BLM}. This solenoid is a set on which the multiplication 
is invertible. It is constructed via a suitable inverse limit 
structure (under iterated multiplication by $(k+\ell)$), starting 
from the $1$-torus (or unit circle), represented by the unit interval 
$[0,1)$ with arithmetic taken mod $1$. Together with the above, 
we obtain the following commutative diagram
\begin{equation} \label{eq:diagram}
   \begin{CD}
    \YY^{\mathrm{TM}}_{k,\ell}   @>\phi>> \YY^{\mathrm{pd}}_{k,\ell} 
        @>\psi>> \mathbb{S}^{}_{k+\ell} \\
    @V \varrho VV @V\varrho^{\ts\prime}VV @VV \times (k+\ell)V \\
    \YY^{\mathrm{TM}}_{k,\ell}   @>\phi>> \YY^{\mathrm{pd}}_{k,\ell} 
        @>\psi>> \mathbb{S}^{}_{k+\ell}
   \end{CD}
\end{equation}
where $\psi$ denotes the torus parametrisation for the generalised
period doubling sequence, in the spirit of \cite{M,BLM,BMS}. The mapping 
$\psi$ is $1$-to-$1$ almost everywhere. Like for the classic period doubling
sequence, it fails to be $1$-to-$1$ on exactly two translation orbits,
which are mapped to a single orbit.

Counting \emph{finite} periodic orbits under the multiplication action
on the solenoid, however, means that the inverse limit is not needed,
so that the corresponding dynamical (or Artin-Mazur) zeta function
coincides with that of the toral endomorphism represented by
multiplication by $m=k+\ell \ge 2$. This, in turn, is given by
\[
     \zeta^{\mathrm{sol}}_{m} (z) \, = \, \frac{1-z}{1-mz}
\]
by an application of \cite[Thm.~1]{BLP}. The number of fixed points
is given by $a^{\mathrm{sol}}_{m} (n) = m^{n} -1$, where $m=k+\ell$
as before for a comparison with \eqref{eq:fix-count}.

\section{Further directions}

A natural question concerns the robustness of the singular continuous
spectrum under simultaneous permutations of positions in $\varrho(1)$
and $\varrho(\bar{1})$, as considered in \cite{Zaks}.  For $k=\ell=1$, the
two possible rules are the Thue-Morse rule $(1\bar{1},\bar{1}1)$ and
its partner $(\bar{1}1,1\bar{1})$, written in obvious shorthand
notation. Since the squares of these two rules are equal, they define
the same hull, and hence the same autocorrelation.

There are three possible rules for $k=2$ and $\ell=1$, namely
$(11\bar{1},\bar{1}\bar{1}1)$, $(1\bar{1}1,\bar{1}1\bar{1})$ and
$(\bar{1}11,1\bar{1}\bar{1})$. The first is our $\varrho^{}_{2,1}$
from above, while the third results in a fixed point that is the
mirror image, and hence possesses the same autocorrelation. The middle
rule, however, has the periodic fixed point $\ldots 1\bar{1} 1\bar{1}
1\bar{1}\ldots$ and hence autocorrelation $\gamma =
(\delta_{0}-\delta_{1}) * \delta_{2\ZZ}$.  Its Fourier transform reads
\[
   \widehat{\gamma} \, = \, \delta_{\ZZ\ts+\frac{1}{2}}\, ,
\]
which is a periodic pure point measure. By general arguments, one can
see that the diffraction measure of the balanced situation (with $1$
and $\bar{1}$ being equally frequent) must be of pure type, and cannot
be a mixture. On the basis of the results from \cite{Nat1,Nat2}, it is
then clear that, given a permutation, the diffraction is either pure
point or purely singular continuous. It remains an interesting
question to decide this explicitly for general $k$ and $\ell$, and a
general permutation.

Moreover, it is clear that similar structures exist in higher
dimensions.  Indeed, starting from the treatment of bijective lattice
substitution systems in \cite{Nat1,Nat2}, it is possible to
demonstrate the singular continuous nature for bijective substitutions
with trivial height lattice and a binary alphabet, and to calculate it
explicitly in terms of Riesz products. Details will be explained in a
forthcoming publication \cite{BG12}.

\section*{Acknowledgement}
We are grateful to Michel Dekking, Xinghua Deng, Natalie Frank and
Daniel Lenz for helpful discussions. This work was supported by the
German Research Council (DFG), within the CRC 701, and by a Leverhulme
Visiting Professorship grant (MB). We thank the Erwin Schr\"{o}dinger
Institute in Vienna for hospitality, where part of the work was done.

\end{document}